\documentclass[prd,showpacs,preprintnumbers,groupedaddress]{revtex4-1} 

\usepackage{amsthm}
\theoremstyle{plain}
\newtheorem{theorem}{Theorem}
\newtheorem{definition}{Definition}

\newtheorem{proposition}{Proposition}

\usepackage{cases}
\usepackage{color}
\usepackage{comment}
\usepackage[normalem]{ulem}
\usepackage{amsfonts}
\usepackage{braket}
\usepackage{mathrsfs}

\usepackage[pdftex]{graphicx}
\usepackage[hang,small,bf]{caption}
\usepackage[subrefformat=parens]{subcaption}
\begin{document}


\title{Photon surfaces as pure tension shells: Uniqueness of thin shell wormholes}


\author{Yasutaka Koga}
\affiliation{Department of Physics, Rikkyo University, Toshima, Tokyo 171-8501, Japan}


\date{\today}

\begin{abstract}
Thin shell wormholes are constructed by joining two asymptotically flat spacetimes along their inner boundaries.
The junction conditions imposed on the spacetimes specify the equation of state of the matter called thin shell distributed along the joined boundaries.
Barcelo and Visser (2000) reported that spherically symmetric thin shell wormholes have their shells, namely the wormhole throats, on the photon spheres if the wormholes are $Z_2$-symmetric across the throats and the shells are of pure tension.
In this paper, first, we consider general {\it joined spacetimes (JSTs)} and show that any {\it $Z_2$-symmetric pure-tensional JST (Z2PTJST)} of $\Lambda$-vacuum has its shell on a photon surface, a generalized object of photon spheres, without assuming any other symmetries.
The class of Z2PTJSTs also includes, for example, brane world models with the shells being the branes we live in.
Second, we investigate the shell stability of Z2PTJSTs by analyzing the stability of the corresponding photon surfaces.
Finally, applying the uniqueness theorem of photon spheres by Cederbaum (2015), we establish the uniqueness theorem of static wormholes of Z2PTJST.
\end{abstract}


\maketitle


\tableofcontents
\newpage

\section{Introduction}
\label{sec:introduction}
Wormholes are spacetimes having two different asymptotic regions and a throat connecting them.
Their structure enables us to travel to another universe.
Holes connecting two regions of an asymptotic region as a shortcut are also called wormholes.
The wormhole solutions to general relativity (GR) and the modified theories of GR have been provided by the many authors~\cite{kanti,kanti2,harko} (see also~\cite{rogatko} and the citation therein).
The uniqueness of wormholes has been proved for Einstein-phantom scalar theory~\cite{yazadjiev,rogatko,rogatko_whmaxwellphantom}.
One of the most important properties of wormholes is that they necessarily violate the energy conditions, which makes it difficult to construct physically reasonable wormhole spacetimes~\cite{morris,hochberg}.
\par
Visser~\cite{visser} proposed a procedure consisting of truncation and gluing of two spacetimes to construct {\it thin shell wormholes}.
By the truncation, inner regions of the two spacetimes are removed and the resulting spacetimes are manifolds with the inner boundaries.
By gluing the two spacetimes along the inner boundaries, we obtain a wormhole spacetime partitioned by the hypersurface, or in other words the throat, into the two regions corresponding to the original two spacetimes.
Through the junction conditions~\cite{israel,textbook:poisson} imposed on the hypersurface, the singularity of the curvature there is interpreted as an infinitesimally thin matter distribution called {\it thin shell}.
\par
Barcelo and Visser~\cite{barcelo} investigated four-dimensional thin shell wormholes consisting of two isometric, static and spherically symmetric spacetimes joined at the same radii and found that the radii of the throats coincide with those of photon spheres, i.e. null circular geodesics.
Subsequently, Kokubu and Harada~\cite{kokubu} extended the analysis to arbitrary dimensions of spacetime and the field equations with the cosmological constant.
From their analysis, we can also find the coincidence of the throats and photon spheres.
The key features of those models are that (1) the wormhole spacetimes are $Z_2$-symmetric across their throats and (2) the shells on the throats have pure tension surface stress energy tensors.
\par
Photon spheres have attracted much attention due to the variety of their applications to astrophysical problems, black hole (BH) shadows~\cite{synge,eht}, quasinormal modes~\cite{cardoso,hod}, and spacetime instability~\cite{cunha,cardoso_psstability}, for example.
Claudel, Virbhadra, and Ellis introduced {\it a photon surface} as the generalization of photon spheres.
It is the geometrical structure which inherits the local properties of photon spheres but does not necessarily have the spherical symmetry.
Photon surfaces have been found in, for example, the accelerated Schwarzschild BH, which is no longer spherically symmetric~\cite{gibbons_2016}.
Uniqueness theorems of spacetimes possessing photon surfaces have been established by the authors~\cite{cederbaum,cederbaum_maxwell,rogatko_psuniqueness,yazadjiev_psuniqueness}.
See also~\cite{koga:psfstability,yoshino_tts,yoshino_dtts,koga3,yazadjiev_psclassification} for investigations concerning photon surfaces.
\par
In this paper, we see that any joined spacetime which is $Z_2$-symmetric across the shell having pure tension surface stress energy tensor has the coincidence shown in~\cite{barcelo} in $\Lambda$-vacuum.
That is, for any $Z_2$-symmetric spacetime joined by a pure tension shell in $\Lambda$-vacuum, the glued inner boundaries of the two original spacetimes must be photon surfaces.
We call the spacetime $Z_2$-symmetric pure-tensional joined spacetime (Z2PTJST).
Not only the thin shell wormholes but also the brane world models by Randall and Sundrum~\cite{randall,randall2} and baby universes~\cite{visser,barcelo} are in the class of Z2PTJST.
The theorem we prove allows us to construct Z2PTJSTs from any $\Lambda$-vacuum spacetime having a photon surface.
\par
This paper is organized as follows.
In Sec.~\ref{sec:joined-spacetime}, we define a {\it joined spacetime} (JST) and review Israel's junction conditions, the field equations the joined spacetime has to satisfy in addition to Einstein equation.
In Sec.~\ref{sec:z2-symmetry}, we define the $Z_2$-symmetry of the joined spacetime and see that the junction conditions reduce to simple forms.
In Sec.~\ref{sec:psf-wht}, we define a {\it pure-tensional joined spacetime} and establish one of our main theorems, the coincidence between pure-tensional shells and photon surfaces.
In Sec.~\ref{sec:stability}, we analyze the stability of the JST against perturbations of the shell preserving the $Z_2$-symmetry and find the stability also coincides with the stability of the corresponding photon surfaces.
In Sec.~\ref{sec:uniqueness}, we establish the uniqueness theorem of $Z_2$-symmetric pure-tensional wormholes applying the uniqueness theorem of photon spheres by Cederbaum~\cite{cederbaum} to the JST.
Sec.~\ref{sec:conclusion} is devoted to the conclusion.
We suppose that manifolds and fields on them are smooth and the spacetimes are $(d+1)$-dimensional with $d\ge2$ if it is not mentioned particularly.

\section{Joined spacetime}
\label{sec:joined-spacetime}
We consider joining two spacetimes $(M_\pm,g_\pm)$ along their inner boundaries $\Sigma_\pm$.
The resulting spacetime $(\mathcal{M},g,\Sigma)$ with the hypersurface $\Sigma$ corresponding to the joined boundaries $\Sigma_\pm$ is called a joined spacetime.
The procedure for constructing the manifold $\mathcal{M}$ consists of {\it truncation} and {\it gluing} of $M_\pm$.
We also introduce the field equations interpreted as Einstein equation for $(\mathcal{M},g,\Sigma)$.
\subsection{Truncation and gluing of manifolds}
Let $M_\pm$ be manifolds with hypersurfaces $\Sigma_\pm$ which partition $M_\pm$ into two regions $M_\pm^{\rm ex}$ and $M_\pm^{\rm in}$ and assume $\Sigma_\pm$ are diffeomorphic.
Truncating $M_\pm$ along $\Sigma_\pm$, we obtain the manifolds $\bar{M}_\pm:=M_\pm\setminus M_\pm^{\rm in}$ with the inner boundaries $\Sigma_\pm$.
Gluing $\bar{M}_\pm$ along $\Sigma_\pm$, i.e. identifying $\Sigma_\pm$ by a diffeomorphism $\psi:\Sigma_+\to\Sigma_-$, we construct a new manifold $\mathcal{M}$~\cite{visser}.
$\mathcal{M}$ is a manifold such that a hypersurface $\Sigma$ partitions it into two regions corresponding to $M_\pm^{\rm ex}$.
\par
The gluing also induces tensor fields on $\mathcal{M}$ from $M_\pm$.
As we see below, we are concerned with {\it a tensor distribution} and {\it jump of tensor fields} on $\mathcal{M}$ across $\Sigma$.
Their values on $\Sigma$ are given by summations of each the tensor fields of $M_\pm$ on $\Sigma_\pm$.
To deal with the summations, we need to specify the diffeomorphism which identifies the tangent bundles $T_{\Sigma_\pm}M_\pm$ of $M_\pm$ on $\Sigma_\pm$.
Since $\psi$ induces the diffeomorphism $\psi^*:T\Sigma_+\to T\Sigma_-$ of the tangent bundles $T\Sigma_\pm$ of $\Sigma_\pm$, it is sufficient to specify the diffeomorphism $\psi_N:N\Sigma_+\to N\Sigma_-$ of the normal bundles $N\Sigma_\pm=T_{\Sigma_\pm}M_\pm/T\Sigma_\pm$ of $\Sigma_\pm$.
Given metrics $g_\pm$ on $M_\pm$, it is natural to require
\begin{equation}
\psi_N:N_+\mapsto N_-
\end{equation}
for the unit normal vector fields $N_\pm\in N\Sigma_\pm$ of $\Sigma_\pm$ which are given so that $g_\pm(N_\pm,N_\pm)=1$ and $N_+$ and $N_-$ points inside $M_+^{\rm ex}$ and $M_-^{\rm in}$, respectively.
The requirement is frequently seen in, for example,~\cite{textbook:poisson,barcelo,kokubu}.
\par
Then, in the current paper, we express the gluing by
\begin{equation}
\mathcal{M}=\bar{M}_+\cup_{\psi,\psi_N}\bar{M}_-,
\end{equation}
which is characterized by the diffeomorphisms $\psi$ and $\psi_N$ above.
Note that $\psi$ and $\psi_N$ are dependent.
The projections of $N\Sigma_\pm$ to $\Sigma\pm$ give $\psi$ from $\psi_N$.
Note also that we have denoted $M_\pm^{\rm ex}$ as the regions we keep for convenience.
We can exchange the roles of the exterior regions $M_\pm^{\rm ex}$ and the interior regions $M_\pm^{\rm in}$ freely.
See Fig~\ref{fig:gluing} for the picture of the construction of $\mathcal{M}$.
\begin{figure}[h]
  \begin{minipage}[c]{1.0\linewidth}
    \centering
    \includegraphics[width=200pt]{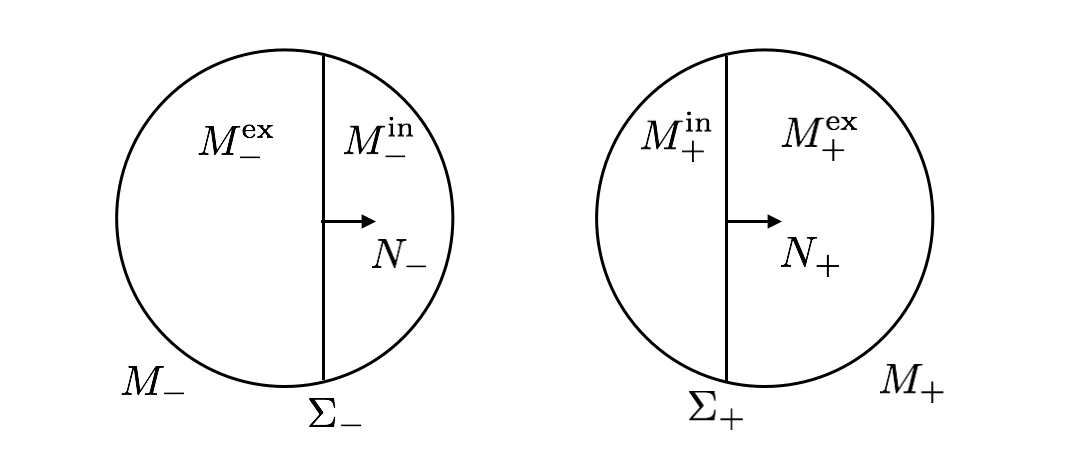}
    \subcaption{Two manifolds $M_\pm$}\label{fig:gluing1}
  \end{minipage}\\
  \begin{minipage}[c]{1.0\linewidth}
    \centering
    \includegraphics[width=200pt]{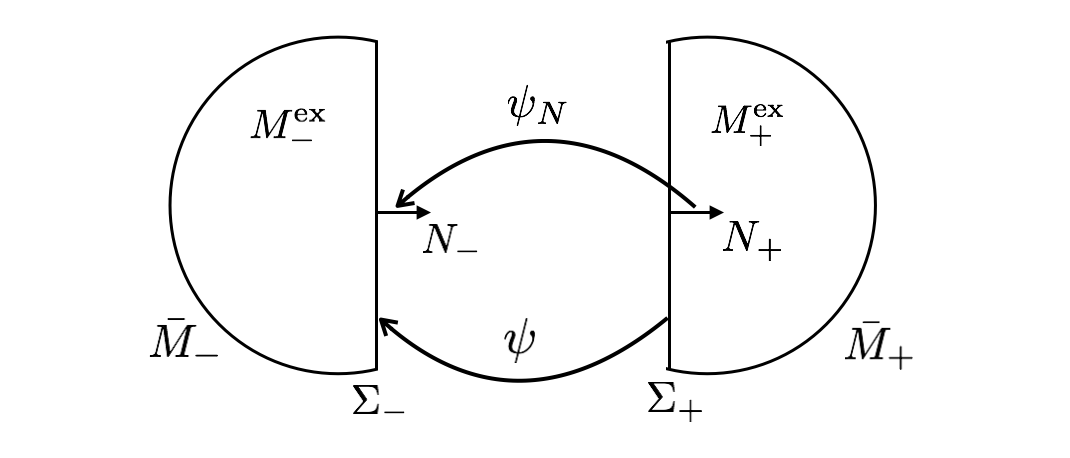}
    \subcaption{Truncation and gluing along $\Sigma_\pm$}\label{fig:gluing2}
  \end{minipage}\\
  \begin{minipage}[c]{1.0\linewidth}
    \centering
    \includegraphics[width=200pt]{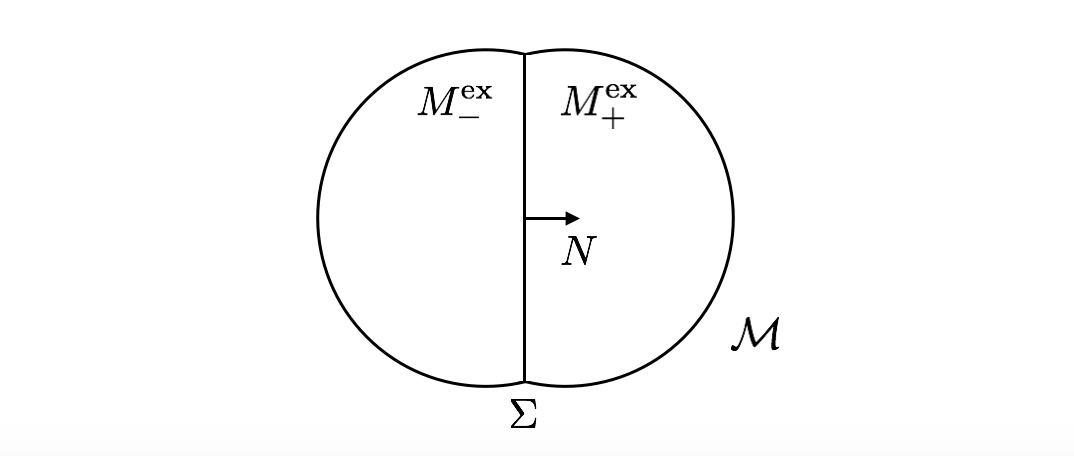}
    \subcaption{Resulting manifold $\mathcal{M}$}\label{fig:gluing3}
  \end{minipage}
  \caption{\subref{fig:gluing1} Two manifolds $M_\pm$ partitioned by the hypersurfaces $\Sigma_\pm$ are \subref{fig:gluing2} truncated and glued along $\Sigma_\pm$ (identified at $\Sigma_\pm$) by $\psi$ and $\psi_N$. \subref{fig:gluing3} The resulting manifold $\mathcal{M}$ possesses the regions $M_\pm^{\rm ex}$ partitioned by $\Sigma$.}\label{fig:gluing}
\end{figure}
\subsection{Tensor distribution}
Let $l$ be a smooth function in the neighborhood of $\Sigma$ in $\mathcal{M}$ satisfying $l=0$ on $\Sigma$, $l>0$ on $M_+$, and $l<0$ on $M_-$.
{\it The tensor distribution $T$} on $\mathcal{M}$ of tensors $T_\pm$ on $M_\pm$ is defined by
 \begin{equation}
T=\Theta(l)T_++\Theta(-l)T_-
\end{equation}
where $\Theta(l)$ is Heaviside distribution,
\begin{equation}
\Theta(l)=\left\{
	\begin{array}{ll}
	1& (l>0)\\
	1/2& (l=0)\\
	0& (l<0).
	\end{array}
\right.
\end{equation}
The metric $g$ on $\mathcal{M}$ is defined as the distribution,
\begin{equation}
g=\Theta(l)g_++\Theta(-l)g_-.
\end{equation}
\subsection{Definition}
According to the discussion above, we define {\it the joined spacetime} constructed from the spacetimes $(M_\pm,g_\pm)$ as follows.
\begin{definition}[Joined spacetime (JST)]
\label{definition:jst}
A triple $(\mathcal{M},g,\Sigma)$ of a manifold $\mathcal{M}$, a metric $g$, and a hypersurface $\Sigma$ is called a joined spacetime constructed from $(M_\pm,g_\pm)$ if the hypersurfaces $\Sigma_\pm$ partitioning $(M_\pm,g_\pm)$ into $M_\pm^{\rm in}$ and $M_\pm^{\rm ex}$ are timelike and
\begin{eqnarray}
\label{eq:connected-sum}
&&\mathcal{M}=\bar{M}_+\cup_{\psi,\psi_N}\bar{M}_-,\\
\label{eq:surface-identification}
&&\Sigma=\Sigma_+\equiv\Sigma_-,\\
&&g=\Theta(l)g_++\Theta(-l)g_-,
\end{eqnarray} 
where the diffeomorphisms
\begin{eqnarray}
\label{eq:surface-identification}
&&\psi:\Sigma_+\to\Sigma_-,\\
\label{eq:normal-identification}
&&\psi_N:N_+\mapsto N_-
\end{eqnarray}
are the identifications of $\Sigma_\pm$ and $N\Sigma_\pm$, respectively.
$\Sigma$ is called the {\it shell} of the joined spacetime.
\end{definition}
\subsection{Einstein equations}
The distribution $g$ of $(\mathcal{M},g,\Sigma)$ may not be smooth across $\Sigma$.
Israel's junction conditions~\cite{israel,textbook:poisson}, which we assume on $g$ on $\Sigma$, are motivated from Einstein equation.
Here we consider the system consisting of Einstein equation and the junction conditions for joined spacetimes.
We call the system {\it Einstein equations}.
\begin{definition}[Einstein equations]
A joined spacetime $(\mathcal{M},g,\Sigma)$ is said to satisfy Einstein equations if $(\bar{M}_\pm,g_\pm)$ satisfy Einstein equation and the first junction condition Eq.~(\ref{eq:first-junction}), the second junction condition Eq.~(\ref{eq:second-junction}), and the equation of motion (EOM) of the shell Eq.~(\ref{eq:eom-brane}) in the following are satisfied on $\Sigma$.
\end{definition}
\subsubsection{First junction condition}
{\it Israel's first junction condition} requires the induced metrics $h_\pm$ of $\Sigma_\pm$ to equal.
It is expressed as,
\begin{equation}
\label{eq:first-junction}
\left[h\right]=0
\end{equation}
where $\left[A\right]$ is the jump of tensor fields $A_\pm$ of $M_\pm$ across $\Sigma$,
\begin{equation}
\left[A\right]:=\left.A_+\right|_{\Sigma_+}-\left.A_-\right|_{\Sigma_-}.
\end{equation}
Note that the condition together with the gluing condition, Eq.~(\ref{eq:normal-identification}), implies
\begin{equation}
\left[g\right]=0.
\end{equation}
Thus, the first junction condition for a joined spacetime $(\mathcal{M},g,\Sigma)$ guarantees the continuity of the metric distribution $g$ across $\Sigma$~\cite{textbook:poisson}.
\subsubsection{Second junction condition}
The distribution $g$ may not be smooth across $\Sigma$ and the curvature, the second derivative of $g$, can be singular there.
The second junction condition is what relates such singularity on $\Sigma$ to the infinitesimally thin matter distribution on the hypersurface.
In the presence of the singular terms in Einstein tensor distribution, Einstein equation leads to
\begin{equation}
\label{eq:second-junction}
-\frac{1}{8\pi}\left(\left[\chi\right]-\left[\theta\right]h\right)=S
\end{equation}
where $\chi_\pm$ is the second fundamental form of $\Sigma_\pm$ in $M_\pm$ with respect to $N_\pm$, $\theta_\pm$ is the trace of $\chi_\pm$, $h$ is the induced metric on $\Sigma$ given by $h(X,Y)=g(X,Y)$ $\forall X,Y\in T\Sigma$, and $S$ is the surface stress energy tensor of the matter on $\Sigma$~\cite{textbook:poisson}.
This is called {\it Israel's second junction condition}.
\par
The corresponding stress energy tensor $T_\Sigma$ of the shell as the matter on $(\mathcal{M},g,\Sigma)$ is given by,
\begin{equation}
T_\Sigma=\delta(l)\Phi^*(S).
\end{equation}
where $\Phi^*:T_p\Sigma\to T_{\Phi(p)}\mathcal{M}$ is the push forward associated with the embedding $\Phi : \Sigma\to \mathcal{M}$ and $\delta(l)$ is the delta function, or Dirac distribution.
As a physical interpretation, $\Sigma$ with $S$ is called {\it a thin shell}.
Eq.~(\ref{eq:second-junction}) is equivalent to the ordinary Einstein equation with the stress energy tensor $T_\Sigma$.
Therefore, the brane world models can be regarded as joined spacetimes~\cite{randall,randall2}.
See~\cite{visser} for the details of the interpretation.
\par
If $S=0$, i.e. $\left[\chi\right]=0$, Christoffel symbols are continuous across $\Sigma$ and Riemann curvature has no singular terms~\cite{textbook:poisson}.
\subsubsection{EOM of the shell}
From Einstein equation, the surface stress energy tensor $S$ of the shell satisfies a conservation law on $\Sigma$ as ordinary matter does in $\mathcal{M}$.
It is given by,
\begin{equation}
\label{eq:eom-brane}
\nabla^h\cdot S+\left[T_N\right]=0,
\end{equation}
where $\nabla^h$ is the covariant derivative associated with $h$, $\nabla^h\cdot S$ expresses ${\nabla^h}^bS_{ab}$, and ${T_N}_\pm$ is given by ${{T_N}_\pm}_a:={T_\pm}_{\mu\nu} e^\mu_aN^\nu_\pm$ with the coordinate basis $\left\{e^\mu_a\right\}$ on $\Sigma$~\cite{textbook:poisson}.

\section{$Z_2$-symmetry of a joined spacetime}
\label{sec:z2-symmetry}
We focus on a $Z_2$-symmetric joined spacetime $(\mathcal{M},g,\Sigma)$ across $\Sigma$.
The $Z_2$-symmetry is also called reflection symmetry across $\Sigma$ in literatures.
\subsection{Definition}
\begin{definition}[$Z_2$-symmetric joined spacetime (Z2JST)]
\label{definition:z2-symmetry}
A joined spacetime $(\mathcal{M},g,\Sigma)$ constructed from $(M_\pm,g_\pm)$ is said to be $Z_2$-symmetric across $\Sigma$ if there exists an isometry $\phi:(M_+,g_+)\to (M_-,g_-)$ such that
\begin{eqnarray}
\label{eq:z2-diffeo-sigma}
&&\left.\phi\right|_{\Sigma_+}:\Sigma_+ \to \Sigma_-,\\
\label{eq:z2-identification}
&&\left.\phi\right|_{\Sigma_+}=\psi,\\
\label{eq:normal-inversion}
&&\phi^*:N_+\mapsto -N_-,
\end{eqnarray}
where $\left.\phi\right|_{\Sigma_+}$ is the restriction of $\phi$ to $\Sigma_+$, $\psi:\Sigma_+\to\Sigma_-$ is the diffeomorphism in Eq.~(\ref{eq:surface-identification}) in Definition~\ref{definition:jst}, and $\phi^*:TM_+\to TM_-$ is the map induced from $\phi$.
\end{definition}
The condition Eq.~(\ref{eq:normal-inversion}) together with Eq.~(\ref{eq:z2-diffeo-sigma}) implies $\phi:M_+^{\rm ex}\to M_-^{\rm ex}:M_+^{\rm in}\to M_-^{\rm in}$.
The picture of the definition is shown in Fig~\ref{fig:z2}.
For the validity of Definition~\ref{definition:z2-symmetry}, see Appendix~\ref{app:z2-coords}.
\begin{figure}[h]
  \begin{minipage}[c]{1.0\linewidth}
    \centering
    \includegraphics[width=200pt]{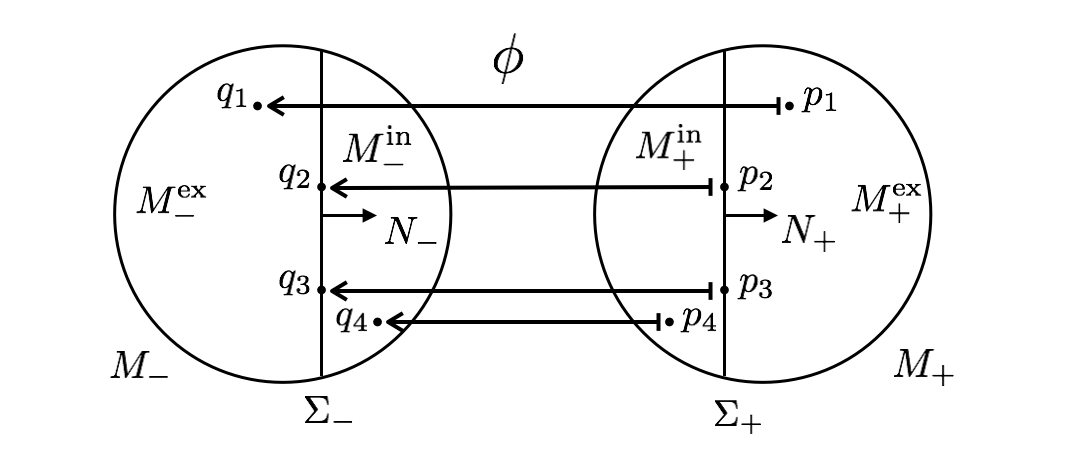}
    \subcaption{Isometry $\phi$}\label{fig:z21}
  \end{minipage}
  \begin{minipage}[c]{1.0\linewidth}
    \centering
    \includegraphics[width=200pt]{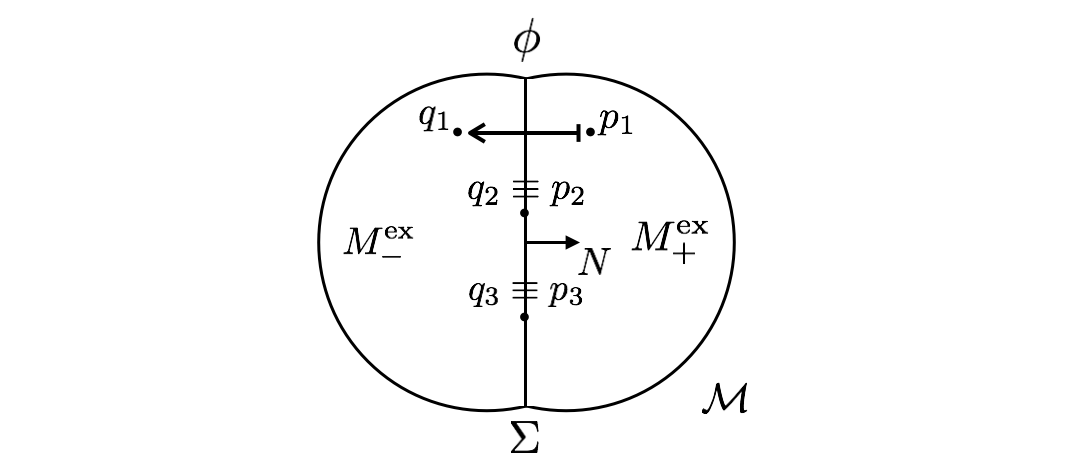}
    \subcaption{$\phi$ acting on $\mathcal{M}$}\label{fig:z22}
  \end{minipage}
  \caption{\subref{fig:z21} Points $p_1,p_2,...\in M_+$ in the regions $M_+^{\rm in}$, $M_+^{\rm ex}$, and $\Sigma_+$ are mapped by the isometry $\phi$ to $q_1,q_2,...\in M_-$ in the corresponding regions $M_-^{\rm in}$, $M_-^{\rm ex}$, and $\Sigma_-$, respectively. \subref{fig:z22} After the truncation and gluing of $M_\pm$, $\phi$ acts on $\mathcal{M}$ as the reflection across $\Sigma$ with the fixed points $p_1\equiv q_1,p_2\equiv q_2\in\Sigma$, which have been identified by $\psi=\phi|_{\Sigma_+}$.}\label{fig:z2}
\end{figure}
\subsection{Junction conditions under the $Z_2$-symmetry}
The $Z_2$-symmetry of a joined spacetime simplifies the junction conditions.
\begin{proposition}
\label{proposition:z2-junction}
Let $(\mathcal{M},g,\Sigma)$ be a Z2JST.
Then, the first junction condition,
\begin{equation}
\left[h\right]=0,
\end{equation}
is satisfied and the second junction condition reduces to
\begin{equation}
\label{eq:z2-junction}
-\frac{1}{4\pi}\left(\chi-\theta h\right)=S,
\end{equation}
where $\chi:=\chi_+=-\chi_-$ and $\theta$ is the trace of $\chi$.
\end{proposition}
\begin{proof}
Let $X,Y\in T\Sigma_\pm$ be arbitrary vectors which are identical under the isometry $\phi$, i.e. the map $\phi^*:TM_+\to TM_-$ induced from $\phi$ maps the vectors as $X\mapsto X$, $Y\mapsto Y$.
Note that the map $\phi^*$ is also regarded as the map of any types of tensors~\cite{textbook:wald}.
The induced metrics $h_\pm(X,Y):=g_\pm(X,Y)$ are mapped as
\begin{equation}
\label{eq:map-h}
h_+(X,Y)=g_+(X,Y)\mapsto g_-(X,Y)=h_-(X,Y)
\end{equation}
by $\phi^*$.
The second fundamental forms $\chi_\pm(X,Y):=(\nabla_\pm n_\pm)(X,Y)$ are mapped as
\begin{equation}
\label{eq:map-chi}
\chi_+(X,Y)=(\nabla_+ n_+)(X,Y)\mapsto (\nabla_-\phi^*(n_+))(X,Y)=-(\nabla_-n_-)(X,Y)=-\chi_-(X,Y).
\end{equation}
by $\phi^*$, where $\nabla_\pm$ are the covariant derivatives associated with $g_\pm$ and $n_\pm=g_\pm(N_\pm,\cdot)$ are the normal $1$-forms dual to $N_\pm$, which are mapped as $\phi^*:n_+=g_+(N_+,\cdot)\mapsto g_-(-N_-,\cdot)=-n_-$ from Eq.~(\ref{eq:normal-inversion}).
Since $h_\pm$ and $\chi_\pm$ are tensors of $\Sigma_\pm$, i.e. tensors taking the values on the space $T\Sigma_\pm\otimes T\Sigma_\pm$, we can regard Eqs.~(\ref{eq:map-h}) and~(\ref{eq:map-chi}) as the mapping induced from $\left.\phi\right|_{\Sigma_+}$ rather than $\phi$ due to Eq.~(\ref{eq:z2-diffeo-sigma}).
That is, the map $(\left.\phi\right|_{\Sigma_+})^*:T\Sigma_+\to T\Sigma_-$ induced from $\left.\phi\right|_{\Sigma_+}$ is a map such that
\begin{eqnarray}
(\left.\phi\right|_{\Sigma_+})^*&:&h_+\mapsto h_-\\
&:&\chi_+\mapsto -\chi_-.
\end{eqnarray}
Therefore, from Eq.~(\ref{eq:z2-identification}), 
we have
\begin{eqnarray}
\psi^*&:&h_+\mapsto h_-\\
&:&\chi_+\mapsto -\chi_-
\end{eqnarray}
for the map $\psi^*:T\Sigma_+\to T\Sigma_-$ induced from $\psi$.
This means $h_+=h_-$ and $\chi_+=-\chi_-$ on $\Sigma$ since $\psi^*$ is the identification.
Then the first junction condition,
\begin{equation}
[h]=h_+-h_-=0,
\end{equation}
is satisfied.
The second fundamental form and its trace satisfy
\begin{eqnarray}
[\chi]&=&\chi_+-\chi_-=2\chi_+,\\ 
\left[\theta\right] &=&\theta_+-\theta_-=2\theta_+.
\end{eqnarray}
Defining $\chi:=\chi_+$ and $\theta:=\mathrm{Tr}(\chi)$, the second junction condition reduces to
\begin{equation}
S=-\frac{1}{8\pi}([\chi]-[\theta]h)=-\frac{1}{4\pi}(\chi-\theta h).
\end{equation}
\end{proof}

\section{Pure-tensional joined spacetime}
\label{sec:psf-wht}
Here, we define {\it a pure-tensional joined spacetime}.
After reviewing photon surfaces, we prove that the shell of the spacetime is a photon surface of the original spacetimes.
\subsection{Definition}
A thin shell having pure-trace stress energy tensor is called {\it a pure tension shell}.
We consider a joined spacetime with a pure tension shell.
\begin{definition}[Pure-tensional joined spacetime (PTJST)]
\label{definition:ptwh}
Let $(\mathcal{M},g,\Sigma)$ be a joined spacetime constructed from $(M_\pm,g_\pm)$.
Let $S$ be the surface stress energy tensor of $\Sigma$ and $T_\pm$ be the stress energy tensors of $(M_\pm,g_\pm)$.
$(\mathcal{M},g,\Sigma)$ is called a pure-tensional joined spacetime (PTJST) if $(\mathcal{M},g,\Sigma)$ satisfies Einstein equations with the conditions,
\begin{eqnarray}
\label{eq:ptwh-puretension}
&&S=-\epsilon h,\\
\label{eq:ptwh-vacuum}
&&T_\pm=-\frac{\Lambda_\pm}{8\pi} g_\pm
\end{eqnarray}
where $\epsilon:\Sigma\to\mathbb{R}$ is the tension of the shell $\Sigma$, $h$ is the induced metric on $\Sigma$, and $\Lambda_\pm$ are the cosmological constants of $(M_\pm,g_\pm)$.
$\Sigma$ is called {\it the pure tension shell} of $(\mathcal{M},g,\Sigma)$.
\end{definition}
\subsection{Photon surface}
A photon surface was first introduced by Claudel {\it et al.}~\cite{claudel} as the generalization of a photon sphere of Schwarzschild spacetime.
The photon surface need not have any global symmetries such as stationarity and spherical symmetry.
It inherits only the local properties of the photon sphere.
\begin{definition}[Photon surface]
\label{definition:photon-surface}
A photon surface of $(M, g)$ is an immersed, nowhere-spacelike
hypersurface $S$ of $(M, g)$ such that, for every point $p\in S$ and every null vector $k\in T_pS$, there exists a null geodesic $\gamma : (-\varepsilon,\varepsilon) \to M$ of $(M, g)$ such that $\dot{\gamma}(0) =k, |\gamma|\subset S$.
\end{definition}
Claudel {\it et al.}~\cite{claudel} and Perlick~\cite{perlick} proved the theorem about the equivalent condition for a timelike hypersurface to be a photon surface.
Photon surfaces are characterized by the second fundamental forms.
\begin{theorem}[Claudel {\it et al.} (2001), Perlick (2005)]
\label{theorem:equivalent-psf}
Let $S$ be a timelike hypersurface of spacetime $(M,g)$ with $d+1:=\dim M\ge3$.
Let $h_{ab}$, $\chi_{ab}$, and $\theta$ be the induced metric, the second fundamental form, the trace of $\chi_{ab}$, respectively.
Then $S$ is a photon surface if and only if it is totally umbilic, i.e.
\begin{equation}
\label{eq:umbilic}
\chi_{ab}=\frac{\theta}{d}h_{ab}\;\;\; \forall p\in S.
\end{equation}
\end{theorem}
\par
The following proposition is necessary for the proof of our theorem.
\begin{proposition}[Photon surfaces in $\Lambda$-vacuum]
\label{proposition:psf-in-vacuum}
Let $(M,g)$ be a $\Lambda$-vacuum spacetime with $\dim M\ge3$.
Then, a timelike photon surface $S$ of $(M,g)$ is constant mean curvature (CMC).
\end{proposition}
\begin{proof}
Let $n^\mu$, $h_{\mu\nu}=g_{\mu\nu}-n_\mu n_\nu$, and $\chi_{\mu\nu}$ be the unit normal vector, the induced metric, and the second fundamental form of $S$, respectively.
From the Codazzi-Mainardi equation, we have
\begin{eqnarray}
R_{\mu\alpha\nu\beta}h^\mu_\rho h^\alpha_\gamma h^\nu_\sigma n^\beta
&=&\nabla^h_\rho\chi_{\gamma\sigma}-\nabla^h_\gamma\chi_{\rho\sigma}\nonumber\\
&=&\frac{1}{d}\left[\left(\nabla^h_\rho\theta\right) h_{\gamma\sigma}-\left(\nabla^h_\gamma\theta\right) h_{\rho\sigma}\right]
\end{eqnarray}
where $d$ is the dimension of $S$, $\theta=h^{\rho\gamma}\chi_{\rho\gamma}$ is the mean curvature, and $\nabla^h_\rho$ is the covariant derivative on $S$ associated with $h_{\rho\gamma}$.
We have used Theorem~\ref{theorem:equivalent-psf} in the last equality.
Contracting with $h^{\gamma\sigma}$, the equation reduces to
\begin{equation}
h^{\gamma\sigma}R_{\mu\alpha\nu\beta}h^\mu_\rho h^\alpha_\gamma h^\nu_\sigma n^\beta
=-R_{\mu \beta}h^\mu_\rho n^\beta
=\frac{d-1}{d}\nabla^h_\rho\theta.
\end{equation}
From that $(M,g)$ is $\Lambda$-vacuum, i.e. $R_{\mu\nu}=[2/(d-1)]\Lambda g_{\mu\nu}$ for some constant $\Lambda$, we have
\begin{equation}
\nabla^h_\rho\theta=0.
\end{equation}
Therefore, $\theta=const.$ along $S$ and $S$ is CMC.
\end{proof}
See also Proposition~3.3 in~\cite{cederbaum} for general totally umbilic hypersurfaces.
\subsection{Photon surfaces as pure tension shells}
The following theorem states about the coincidence of pure-tensional shells of joined spacetimes and photon surfaces.
\begin{theorem}[Photon surface as a pure tension shell]
\label{theorem:psf-wh}
Let $(\mathcal{M},g,\Sigma)$ be a Z2JST constructed from $(M_\pm,g_\pm)$ with $d+1:=\dim \mathcal{M}=\dim M_\pm\ge3$.
Then $(\mathcal{M},g,\Sigma)$ is a PTJST if and only if $(\bar{M}_\pm,g_\pm)$ are $\Lambda$-vacuum and $\Sigma_\pm$ are photon surfaces.
The tension $\epsilon$ and the mean curvature $\theta$ of the shell $\Sigma$ are constant and given by the relation $\theta=\pm\theta_\pm=-4\pi[d/(d-1)] \epsilon$.
\end{theorem}
\begin{proof}
From the $Z_2$-symmetry of $(\mathcal{M},g,\Sigma)$ and Proposition~\ref{proposition:z2-junction}, the first junction condition is automatically satisfied and the second junction condition reduces to Eq.~(\ref{eq:z2-junction}) where $\chi:=\chi_+=-\chi_-$ and $\theta:=\mathrm{Tr}(\chi)=\theta_+=-\theta_-$.
Then the Einstein equations for $(\mathcal{M},g,\Sigma)$ consist of the reduced second junction condition, Eq.~(\ref{eq:z2-junction}), the EOM of the shell $\Sigma$, Eq.~(\ref{eq:eom-brane}), and Einstein equation on $(\bar{M}_\pm,g_\pm)$.
The induced metric on $\Sigma$ is given by $h=h_+=h_-$ due to the fact that the first junction condition is satisfied.
In fact, $h(X,Y):=g(X,Y)=\frac{1}{2}(g_+(X,Y)+g_-(X,Y))=\frac{1}{2}(h_+(X,Y)+h_-(X,Y))=h_+(X,Y)=h_-(X,Y)$ $\forall X,Y\in T\Sigma\equiv T\Sigma_\pm$.
\par
We prove the ``if" part.
From that $(\bar{M}_\pm,g_\pm)$ are isometric and $\Lambda$-vacuum, the energy momentum tensors satisfy $T_\pm=-(\Lambda/8\pi)g_\pm$ for the common cosmological constant $\Lambda$ and Eq.~(\ref{eq:ptwh-vacuum}) in 
Definition~\ref{definition:ptwh} is satisfied.
From Theorem~\ref{theorem:equivalent-psf}, the photon surfaces $\Sigma_\pm$ give the conditions
\begin{equation}
\chi=\pm\chi_\pm=\frac{\pm\theta_\pm}{d}h_\pm=\frac{\theta}{d}h.
\end{equation}
From Proposition~\ref{proposition:psf-in-vacuum}, we have
\begin{equation}
\label{eq:theta-equals-const}
\theta=\pm\theta_\pm=const.
\end{equation}
The second junction condition Eq.~(\ref{eq:z2-junction}) then reduces to
\begin{equation}
\frac{1}{4\pi}\frac{d-1}{d}\theta h=S.
\end{equation}
Letting $\epsilon$ be a function on $\Sigma$ given by
\begin{equation}
\label{eq:pressure-mean-curvature}
\epsilon=-\frac{1}{4\pi}\frac{d-1}{d}\theta,
\end{equation}
the surface stress energy tensor becomes
\begin{equation}
S=-\epsilon h
\end{equation}
and Eq.~(\ref{eq:ptwh-puretension}) in Definition~\ref{definition:ptwh} is satisfied.
From Eqs.~(\ref{eq:theta-equals-const}) and~(\ref{eq:pressure-mean-curvature}), we have $\nabla^h \epsilon=0$ implying
\begin{equation}
\nabla^h\cdot S=0.
\end{equation}
Therefore, from the fact that $T_\pm=-(\Lambda/8\pi)g_\pm$, the EOM of the shell $\Sigma$, Eq.~(\ref{eq:eom-brane}), 
\begin{equation}
\nabla^h\cdot S+\left[T_N\right]=0,
\end{equation}
is satisfied.
$(\mathcal{M},g,\Sigma)$ is a joined spacetime satisfying Definition~\ref{definition:ptwh} of PTJST.
\par
We prove the ``only if" part.
From Definition~\ref{definition:ptwh}, the PTJST $(\mathcal{M},g,\Sigma)$ satisfies $S=-\epsilon h$ and $T_\pm=-(\Lambda/8\pi)g_\pm$.
Then $(\bar{M}_\pm,g_\pm)$ are $\Lambda$-vacuum and the second junction condition, Eq.(\ref{eq:z2-junction}), under $Z_2$-symmetry requires that
\begin{equation}
\chi=\chi_+=-\chi_-\propto h,
\end{equation}
i.e. $\Sigma_\pm$ are timelike totally umbilic hypersurfaces.
From Theorem~\ref{theorem:equivalent-psf}, $\Sigma_\pm$ are photon surfaces of $(\bar{M}_\pm,g_\pm)$.
\par
From Eqs.~(\ref{eq:theta-equals-const}) and~(\ref{eq:pressure-mean-curvature}), we finally obtain
\begin{equation}
\theta=\pm\theta=-4\pi\frac{d}{d-1}\epsilon=const.
\end{equation}
for the PTJST.
\end{proof}
Theorem~\ref{theorem:psf-wh} applies to the static and dynamical cases of the wormholes investigated in~\cite{barcelo,kokubu}.
From the viewpoint of symmetries of the photon surfaces, the throats of the static cases correspond to $\mathbb{R}\times SO(d)$, $\mathbb{R}\times E(d-1)$, and $\mathbb{R}\times SO(1,d-1)$-invariant photon surfaces while the dynamical cases correspond to $SO(d-1)$, $E(d)$, and $SO(d-2,1)$-invariant photon surfaces in the $\mathbb{R}\times SO(d)$, $\mathbb{R}\times E(d-1)$, and $\mathbb{R}\times SO(d-1,1)$-invariant spacetimes, respectively.
We can see that the spherical and nonspherical throats are photon surfaces from~\cite{koga3}.
\par
The theorem states that any $\Lambda$-vacuum spacetimes with photon surfaces found in literatures are joined to give Z2PTJST.
For example, we can join two Minkowski spacetimes nontrivially to give a Z2PTJST by gluing them at the timelike photon surfaces shown in Examples~1--3 in~\cite{claudel}.
The accelerated black holes also have photon surfaces, which correspond to the photon spheres in zero acceleration~\cite{gibbons_2016} and thus, we can construct ``accelerated wormholes" and ``accelerated baby universes" from the spacetimes.
For less symmetric $\Lambda$-vacuum spacetimes possessing photon surfaces, see~\cite{koga:nonsym-psf}.
\par
We can also confirm that the positive and negative branes in~\cite{randall,randall2} are indeed located on the four-dimensional timelike photon surfaces of the five-dimensional bulk spacetime from the facts that the bulk spacetime is conformally transformed Minkowski spacetime and photon surfaces are invariant manifolds under conformal transformations~\cite{claudel}.

\section{Stability of pure-tensional joined spacetime}
\label{sec:stability}
We consider a shell perturbation of Z2JST preserving its $Z_2$-symmetry.
That is, perturbing the hypersurfaces $\Sigma_\pm$ of $(M_\pm,g_\pm)$ to $\widetilde{\Sigma}_\pm$, we rejoin the spacetimes along the new boundaries $\widetilde{\Sigma}_\pm$ to give a new perturbed Z2JST $(\widetilde{\mathcal{M}},g,\widetilde{\Sigma})$ (Fig.~\ref{fig:perturbation}).
Since $(M_\pm,g_\pm)$ are isometric to each other including their hypersurfaces $\Sigma_\pm$ and $\widetilde{\Sigma}_\pm$, it is sufficient to focus only on the ``plus" ones and we denote them as $(M,g)$, $\Sigma$, $\widetilde{\Sigma}$, and so on in the following.
\par
After reviewing the deformation formalism of surfaces given by Capovilla and Guven~\cite{capovilla}, we apply it to Z2PTJSTs, that is, the case where $\Sigma$ is a photon surface.
We also see that the stability of $\Sigma$ corresponds to {\it the stability of null geodesics on $\Sigma$} defined in~\cite{koga:psfstability}.
\begin{figure}[h]
  \begin{minipage}[c]{1.0\linewidth}
    \centering
    \includegraphics[width=200pt]{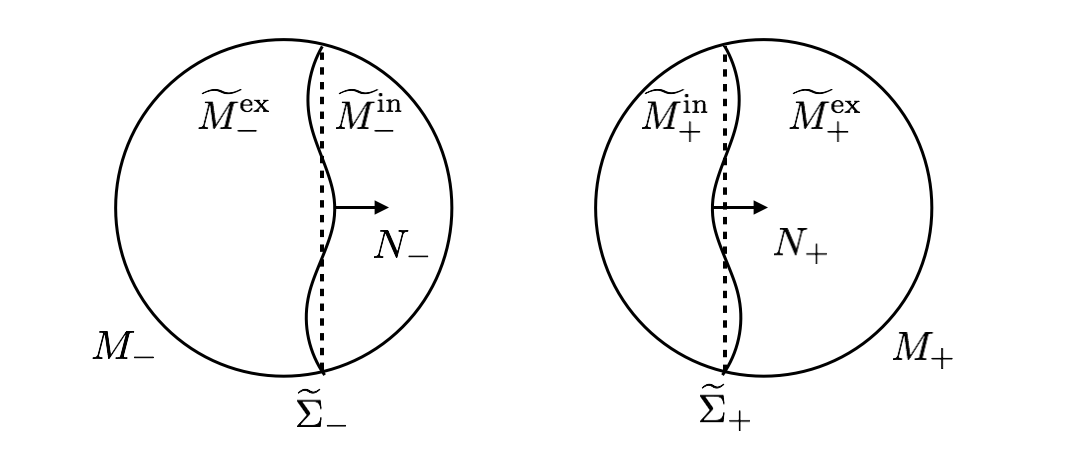}
    \subcaption{Perturbed inner boundaries $\widetilde{\Sigma}_\pm$}\label{fig:perturbation1}
  \end{minipage}\\
  \begin{minipage}[c]{1.0\linewidth}
    \centering
    \includegraphics[width=200pt]{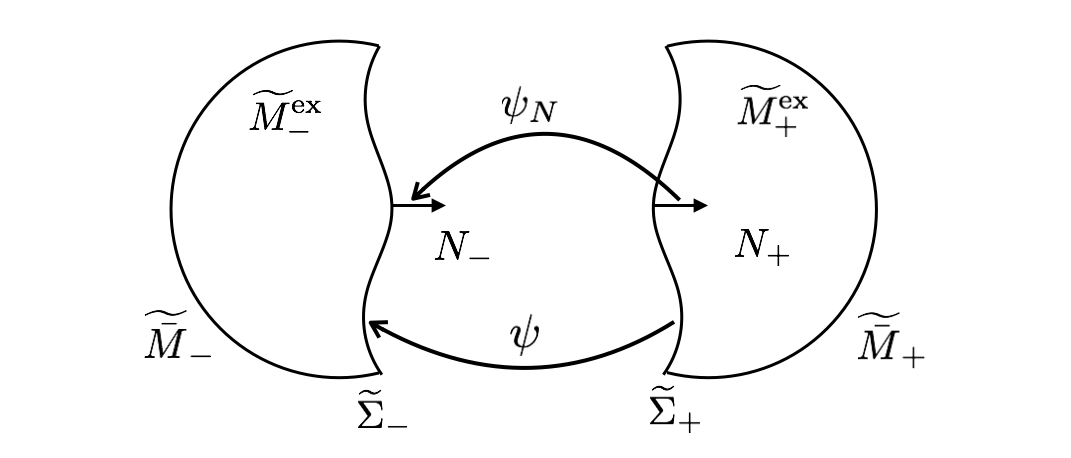}
    \subcaption{Rejoining along $\widetilde{\Sigma}_\pm$}\label{fig:perturbation2}
  \end{minipage}
  \caption{\subref{fig:perturbation1} The unperturbed inner boundaries $\Sigma_\pm$ (dashed lines) are perturbed to $\widetilde{\Sigma}_\pm$ (solid lines). \subref{fig:perturbation2} The truncation and gluing of $M_\pm$ along $\widetilde{\Sigma}_\pm$ gives the perturbed joined manifold $\widetilde{\mathcal{M}}$ with the perturbed shell $\widetilde{\Sigma}:=\widetilde{\Sigma}_+\equiv\widetilde{\Sigma}_-$.}\label{fig:perturbation}
\end{figure}
\subsection{Deformation of a hypersurface}
\label{sec:defo-hypersurface}
Let $\Sigma$ be a timelike hypersurface of a spacetime $(M,g)$.
Consider a one-parameter family of the deformation,
\begin{equation}
\mathcal{F}_\varepsilon\ :\ \Sigma\to\Sigma_\varepsilon, 
\end{equation}
with the parameter $\varepsilon$ where $\Sigma_0=\Sigma$.
Expanding in $\varepsilon$, points on the surfaces are expressed as
\begin{equation}
\label{eq:expansion-epsilon}
\Sigma_\varepsilon^\mu=\Sigma^\mu+\varepsilon X^\mu+\mathcal{O}\left(\varepsilon^2\right)
\end{equation}
in a coordinate system $\left\{x^\mu\right\}$.
The vector field $X$ on $\Sigma$ is called {\it the deviation} between $\Sigma$ and its infinitesimal deformation $\lim_{\varepsilon\to 0}\Sigma_\varepsilon$.
Without loss of generality, we assume $X$ to be orthogonal to $\Sigma$ by diffeomorphisms on $\Sigma$ and $\Sigma_\varepsilon$~\cite{capovilla}.
Define a scalar field $\Phi$ on $\Sigma$ by
\begin{equation}
\label{eq:definition-Phi}
X=\Phi N
\end{equation}
where $N$ is the unit normal vector of $\Sigma$.
This quantity represents the {\it distance} between $\Sigma$ and its infinitesimal deformation.
According to the deformation formalism by Capovilla and Guven~\cite{capovilla}, we have equations which relate $\Phi$ with the geometrical quantities on $\Sigma$.
A tiny fraction of the calculation processes in~\cite{capovilla} is incorrect and we recalculate it for our purpose following their procedure in Appendix~\ref{app:capovilla-recalculation}.
\par
The deformation of the intrinsic geometry of $\Sigma$ is related to $\Phi$ by
\begin{equation}
\label{eq:ricci-derivative}
\left(\nabla_{N}\mathcal{R}\right)\Phi=-2\mathcal{R}_{ab}\chi^{ab}\Phi+2\nabla^h_a\nabla^h_b\left(\chi^{ab}\Phi\right)-2\Delta^h\left(\theta\Phi\right)
\end{equation}
where $\Delta^h:=h^{ab}\nabla^h_a\nabla^h_b$ and $\mathcal{R}_{ab}$ and $\mathcal{R}$ are the Ricci tensor and scalar of $\Sigma_\varepsilon$, respectively.
This is the codimension one version of Eq.~(\ref{eq:deformation-ricci-scalar}) in Appendix~\ref{app:capovilla-recalculation}.
The indices $a,b,...$ are with respect to the coordinate basis vectors $\left\{e_a\right\}$ of $\Sigma$.
For example, a quantity $A_{ab}$ is the scalar resulting from the contraction $A(e_a,e_b)$ of a tensor $A$ acting on $T\Sigma\otimes T\Sigma$.
We  lower and raise the indices by $h_{ab}=h(e_a,e_b)$ and its inverse matrix $h^{ab}$, respectively.
For the deformation of the extrinsic geometry, $\Phi$ obeys
\begin{equation}
\label{eq:Phi-eom}
\nabla^h_a\nabla^h_b\Phi=-\left[\nabla_N\chi_{ab}-\chi_{ac}{\chi^{c}}_b+R_{aNbN}\right]\Phi
\end{equation}
where
the subscripts of $N$ represent the contraction with it, $R_{aNbN}=R_{\mu\alpha\nu\beta}e_a^\mu N^\alpha e_b^\nu N^\beta$.
We obtain the equation from Eq.~(4.6) in~\cite{capovilla} by setting the codimension one.
Note that $\chi_{ab}$ and $\mathcal{R}_{ab}$ is now defined on each surface $\Sigma_\varepsilon$ and can be differentiated along $N$.
The equation,
\begin{equation}
\label{eq:deformation-induced-metric}
2\chi_{ab}=\nabla_N h_{ab},
\end{equation}
is also useful in the calculations.
See Eq.~(3.8) in~\cite{capovilla} for the derivation.
\par
In the case where each surface $\Sigma_\varepsilon$ obeys the same equations of motion, the deformation $\mathcal{F}_\varepsilon$ should be called {\it a perturbation}.
Then we say {\it $\Sigma$ is stable} against the perturbation if $\Phi$ is bounded as time evolves along $\Sigma$ and otherwise {\it unstable}.
\subsection{Perturbation of a photon surface}
\label{sec:pert-psf}
We consider the $Z_2$-symmetric perturbation of a Z2PTJST $(\mathcal{M},g,\Sigma)$.
The perturbed spacetime $(\widetilde{\mathcal{M}},g,\widetilde{\Sigma})$ is also Z2PTJST and therefore, from Theorem~\ref{theorem:psf-wh}, the hypersurface $\widetilde{\Sigma}$ is another photon surface of the original spacetime $(M_\pm,g_\pm)$.
More precisely, $\Sigma$ and $\widetilde{\Sigma}$ corresponds to CMC photon surfaces $\Sigma_\pm$ and $\widetilde{\Sigma}_\pm$ of the $\Lambda$-vacuum spacetime $(M_\pm,g_\pm)$, respectively, because of Theorem~\ref{theorem:psf-wh} and Proposition~\ref{proposition:psf-in-vacuum}.
Then, what we do in the following is to impose the conditions, CMC, totally umbilic (recall Theorem~\ref{theorem:equivalent-psf}), and $\Lambda$-vacuum on the deformation formalism, Eqs.~(\ref{eq:ricci-derivative}) and~(\ref{eq:Phi-eom}).
\par
Let $\Sigma$ be  a CMC photon surface of a $\Lambda$-vacuum spacetime $(M,g)$.
The EOMs of $\Sigma$ are given by
\begin{eqnarray}
\label{eq:umbilic-condition}
\chi_{ab}&=&\frac{\theta}{d}h_{ab},\\
\label{eq:cmc-condition}
\theta&=&const.
\end{eqnarray}
with the $\Lambda$-vacuum condition on the spacetime,
\begin{equation}
\label{eq:vacuum-condition}
G_{\mu\nu}=-\Lambda g_{\mu\nu},
\end{equation}
where $G_{\mu\nu}$ is the Einstein tensor.
We denote the perturbation of $\Sigma$ to $\Sigma_\varepsilon$, which also obeys the EOMs, as $\mathcal{F}^{\rm PS}_\varepsilon:\Sigma\to\Sigma_\varepsilon$.
\par
First we calculate Eq.~(\ref{eq:ricci-derivative}) for $\mathcal{F}^{\rm PS}_\varepsilon$.
The contracted Gauss-Codazzi relation in $\Lambda$-vacuum reads
\begin{equation}
\mathcal{R}+\chi^{ab}\chi_{ab}-\theta^2=-2G_{NN}=2\Lambda.
\end{equation}
From Eq.~(\ref{eq:umbilic-condition}), the equation reduces to
\begin{equation}
\label{eq:ricci-theta}
\mathcal{R}=\frac{d-1}{d}\theta^2+2\Lambda.
\end{equation}
Then the LHS of Eq.~(\ref{eq:ricci-derivative}) reduces to
\begin{equation}
\left(\nabla_N\mathcal{R}\right)\Phi=2\frac{d-1}{d}\left(\theta\nabla_N\theta\right)\Phi.
\end{equation}
The RHS of Eq.~(\ref{eq:ricci-derivative}) reduces to
\begin{eqnarray}
&&-2\chi_{ab}\mathcal{R}^{ab}\Phi+2\nabla^h_a\nabla^h_b\left(\chi^{ab}\Phi\right)-2\Delta^h\left(\theta\Phi\right)\nonumber\\
&=&-2\frac{\theta}{d}\mathcal{R}\Phi+2\nabla^h_a\nabla^h_b\left(\frac{\theta}{d}h^{ab}\Phi\right)-2\Delta^h\left(\theta\Phi\right)\nonumber\\
&=&-2\frac{\theta}{d}\mathcal{R}\Phi-2\frac{d-1}{d}\theta\Delta^h\Phi
\end{eqnarray}
where we have used Eqs.~(\ref{eq:umbilic-condition}) and~(\ref{eq:cmc-condition}) in the first and last equalities, respectively.
Equating both sides with the use of Eq.~(\ref{eq:ricci-theta}), we finally obtain
\begin{equation}
\label{eq:nabla-n-theta}
\nabla_N\theta=-\frac{\theta^2}{d}-\frac{2}{d-1}\Lambda-\Phi^{-1}\Delta^h\Phi.
\end{equation}
\par
Next we calculate Eq.~(\ref{eq:Phi-eom}) for $\mathcal{F}^{\rm PS}_\varepsilon$.
From Eq.~(\ref{eq:umbilic-condition}), the first term in the brackets in Eq.~(\ref{eq:Phi-eom}) reduces to
\begin{eqnarray}
\nabla_N\chi_{ab}
\label{eq:derivative-n-chi}
&=&\nabla_N\left(\frac{\theta}{d}h_{ab}\right)\nonumber\\
&=&\frac{1}{d}\left[\frac{\theta^2}{d}-\frac{2}{d-1}\Lambda -\Phi^{-1}\Delta^h\Phi\right]h_{ab}
\end{eqnarray}
where Eqs.~(\ref{eq:deformation-induced-metric}) and~(\ref{eq:nabla-n-theta}) were used in the last equality.
Then the LHS of Eq.~(\ref{eq:Phi-eom}) reduces to
\begin{eqnarray}
&&-\left[\nabla_N\chi_{ab}-\chi_{ac}{\chi^{c}}_b+R_{aNbN}\right]\Phi\nonumber\\
&=&\frac{1}{d}\Delta^h\Phi h_{ab}+\left[-R_{aNbN}+\frac{2}{d(d-1)}\Lambda h_{ab}\right]\Phi
\end{eqnarray}
where Eqs.~(\ref{eq:umbilic-condition}) and~(\ref{eq:derivative-n-chi}) were used.
Substituting the result into Eq.~(\ref{eq:Phi-eom}), we obtain
\begin{equation}
\nabla^h_a\nabla^h_b\Phi-\frac{1}{d}\Delta^h\Phi h_{ab}=\left[-R_{aNbN}+\frac{2}{d(d-1)}\Lambda h_{ab}\right]\Phi,
\end{equation}
or from the $\Lambda$-vacuum condition Eq.~(\ref{eq:vacuum-condition}),
\begin{equation}
\label{eq:cmcpsf-perturbation}
\nabla^h_a\nabla^h_b\Phi-\frac{1}{d}\Delta^h\Phi h_{ab}=-C_{aNbN}\Phi.
\end{equation}
\par
We have another expression of Eq.~(\ref{eq:cmcpsf-perturbation}).
The contraction of Gauss-Codazzi relation,
\begin{equation}
R_{\alpha\mu\beta\nu}h^\alpha_ah^\mu_ch^\beta_bh^\nu_d=\mathcal{R}_{acbd}-\chi_{ab}\chi_{cd}+\chi_{ad}\chi_{cb},\nonumber
\end{equation}
with $h^{cd}$ gives
\begin{equation}
R_{\alpha\beta}h^\alpha_ah^\beta_b-R_{aNbN}=\mathcal{R}_{ab}-\theta\chi_{ab}+\chi_{ac}\chi^c{}_b.
\end{equation}
from $h^{cd}h^\mu_c h^\nu_d=h^{\mu\nu}=g^{\mu\nu}-N^\mu N^\nu$.
Using the fact that the spacetime is $\Lambda$-vacuum and the hypersurface is totally umbilic, the equation reduces to
\begin{equation}
-C_{aNbN}+\frac{2}{d}\Lambda h_{ab}=\mathcal{R}_{ab}-\frac{d-1}{d^2}\theta^2h_{ab}.
\end{equation}
From Eq.~(\ref{eq:ricci-theta}), we have
\begin{equation}
-C_{aNbN}=\mathcal{R}_{ab}-\frac{\mathcal{R}}{d}h_{ab}.
\end{equation}
Then Eq.~(\ref{eq:cmcpsf-perturbation}) is rewritten as
\begin{equation}
\label{eq:stability-another}
\nabla^h_a\nabla^h_b\Phi-\frac{1}{d}\Delta^h\Phi h_{ab}=\left(\mathcal{R}_{ab}-\frac{\mathcal{R}}{d}h_{ab}\right)\Phi.
\end{equation}
The expression tells us that the linear perturbation $\Phi$ is governed only by the intrinsic geometry, the Ricci curvature $\mathcal{R}_{ab}$, of $(\Sigma,h)$.
\par
Note the master equation of the perturbation, Eq.~(\ref{eq:cmcpsf-perturbation}) or~(\ref{eq:stability-another}), does not have its trace part. Therefore, $\Phi^{-1}\Delta^h\Phi$ is unspecified {\it a priori} and will be determined after we solve the trace-free part of the equation for given initial values of $\Phi$.
\subsection{Stability of the shell and a photon surface}
Let us physically interpret the master equation Eq.~(\ref{eq:cmcpsf-perturbation}) for the perturbation $\Phi$ of the photon surface $\Sigma$.
\par
Let $k_p\in T_p\Sigma$ be a null vector at a point $p\in \Sigma$.
Since $\Sigma$ is a photon surface, there always exists a null geodesic $\gamma(\lambda)$ everywhere tangent to $\Sigma$ such that $k_p=\dot{\gamma}(0)$. 
The contraction of Eq.~(\ref{eq:cmcpsf-perturbation}) with $k=\dot{\gamma}(\lambda)$  gives
\begin{equation}
\nabla^h_k\nabla^h_k\Phi=-C_{kNkN}\Phi.
\end{equation}
Therefore, we have
\begin{equation}
\label{eq:along-lambda}
\frac{d^2}{d\lambda^2}\Phi(\lambda)=-C_{kNkN}\Phi(\lambda)
\end{equation}
along $\gamma(\lambda)$.
This implies that given the initial values $\left.\Phi\right|_{t=t_0}$ and $\left.\partial_a\Phi\right|_{t=t_0}$ at time $t=t_0$, the value of $\Phi$ at the point $p$ in the future is determined by integrating the ordinary differential equation Eq.~(\ref{eq:along-lambda}) along $\gamma(\lambda)$ from $q=\gamma(\lambda_{t=t_0})$ to $p$.
\par
The factor $-C_{kNkN}$ in Eq.~(\ref{eq:along-lambda}) coincides with what determines the {\it stability of the null geodesic $\gamma$ against the perturbation orthogonal to $\Sigma$}~\cite{koga:psfstability}.
That is, an orthogonally perturbed null geodesic $\tilde{\gamma}$ from $\gamma$ is stable (attracted to $\Sigma$) if $-C_{kNkN}<0$ and unstable (repelled from $\Sigma$) if $-C_{kNkN}>0$ in our $\Lambda$-vacuum case (see Proposition~2 in~\cite{koga:psfstability}).
In particular, if $\Sigma$ is {\it a strictly stable photon surface}, i.e. all the null geodesics on $\Sigma$ are stable, the factor in Eq.~(\ref{eq:along-lambda}) is always negative along any null geodesics.
This is suggestive that $\Phi$ will be bounded as time evolves.
In fact, in the case where $-C_{kNkN}$ varies sufficiently slowly, $\Phi$ oscillates with the almost constant amplitude along any null geodesic and will be bounded.
This is quite natural from the physical point of view.
Since a photon surface is a hypersurface generated by null geodesics, the photon surface $\Sigma_\varepsilon$ perturbed from $\Sigma$ is generated by the orthogonally perturbed null geodesics $\tilde{\gamma}$.
Therefore, if all the perturbed null geodesics $\tilde{\gamma}$ are attracted to $\Sigma$, $\Sigma_\varepsilon$ should also be attracted to $\Sigma$ and is stable.
\par
Then our conclusion is this: if $\Sigma$ is a strictly stable photon surface, then $\Sigma$ is stable against the linear perturbation given by $\mathcal{F}^{\rm PS}_\varepsilon:\Sigma\to\Sigma_\varepsilon$.
For Z2PTJSTs, we also conclude as this: if $\Sigma_\pm$ is a strictly stable photon surface of $(M_\pm,g_\pm)$, then the Z2PTJST $(\mathcal{M},g,\Sigma)$ constructed from $(M_\pm,g_\pm)$ is stable against shell perturbations preserving the $Z_2$-symmetry.
\par
In Appendix~\ref{app:example}, we explicitly solve the perturbation equation, Eq.~(\ref{eq:cmcpsf-perturbation}), in the case where the geometry of $\Sigma$ is static and spherically, planar, and hyperbolically symmetric.
The case is frequently seen in, for example,~\cite{visser,barcelo,kokubu}.
The induced metric is given by
\begin{equation}
\label{eq:static-const-curvature-surface-main}
h=a^2(-dt^2+\sigma)\nonumber
\end{equation}
where $a$ is constant and $\sigma$ represents the metric of the $(d-1)$-dimensional space of constant curvature $\alpha=\pm1,0$.
The general solution in spherically ($\alpha=1$) and hyperbolically ($\alpha=-1$) symmetric case is given by
\begin{equation}
\label{eq:general-pert-alpha-nzero}
\Phi=Ce^{\sqrt{\alpha}t}+Fe^{-\sqrt{\alpha}t}\;\;\;(\alpha=\pm1)
\end{equation}
with the arbitrary constants $C$ and $F$ from Eq.~(\ref{eq:Phi-solution-ex}).
In the planar case ($\alpha=0$), we have
\begin{equation}
\Phi=D\eta_{ab}x^ax^b+B_ax^a+C\;\;\;(\alpha=0)
\end{equation}
with the arbitrary constants $D$, $B_a$, and $C$ from Eq.~(\ref{eq:Phi-solution-ex-flat}).
$x^a=(t,x^i)$ is the Cartesian coordinate on $\Sigma$.
Since the Weyl curvature gives $-C_{kNkN}=\mathcal{R}_{kk}=\alpha$ for any null vector $k\in T\Sigma$ with an appropriate scaling, $\Sigma$ is a strictly stable photon surface if and only if $\alpha=-1$, i.e. the hyperbolic case, and it is only the case where the solution $\Phi$ is bounded for any possible perturbation.
Therefore, the result agrees with the above conclusion, indeed.
See Appendix~\ref{app:example} for the derivation and the detailed interpretations of the result.

\section{Uniqueness of pure-tensional wormholes}
\label{sec:uniqueness}
A joined spacetime $(\mathcal{M},g,\Sigma)$ is a wormhole spacetime if both regions $\bar M_\pm$ have asymptotically flat domains.
If, additionally, $(\mathcal{M},g,\Sigma)$ is $Z_2$-symmetric and pure-tensional, the spacetimes $(\bar{M}_\pm,g_\pm)$ constituting $(\mathcal{M},g,\Sigma)$ are asymptotically flat spacetimes with their inner boundaries $\Sigma_\pm$ being photon surfaces according to Theorem~\ref{theorem:psf-wh}.
Cederbaum established the uniqueness theorem of such spacetimes, $(\bar{M}_\pm,g_\pm)$, with some assumptions~\cite{cederbaum}.
Using the theorem, we prove the uniqueness theorem of pure-tensional thin shell wormholes.
\par
In Cederbaum's uniqueness theorem, the spacetime is assumed to be AF-geometrostatic (``AF" stands for ``asymptotically flat")(see below) and the lapse function with respect to the static Killing vector is assumed to be constant along the photon surface being the inner boundary.
In the following, we call the photon surface under the assumptions {\it Cederbaum's photon sphere} (Definition~2.6 in~\cite{cederbaum}).
\subsection{Pure-tensional wormhole}
An {\it AF-geometrostatic spacetime} (Definition~2.1 in~\cite{cederbaum}) is a spacetime which is static, asymptotically flat, and a vacuum solution to Einstein equation with the cosmological constant $\Lambda=0$.
We define the following AF-geometrostatic wormhole spacetime:
\begin{definition}[Static pure-tensional wormhole]
\label{definition:af-ptwh}
Let $(\mathcal{M},g,\Sigma)$ be a Z2PTJST constructed from $(M_\pm,g_\pm)$.
$(\mathcal{M},g,\Sigma)$ is called a static pure-tensional wormhole if $(\bar{M}_\pm,g_\pm)$ are AF-geometrostatic spacetimes and their inner boundaries $\Sigma_\pm$ are static.
\end{definition}
\subsection{Proof of the uniqueness}
We impose a technical assumption on a static pure-tensional wormhole $(\mathcal{M},g,\Sigma)$ to prove the uniqueness theorem.
Since $(\mathcal{M},g,\Sigma)$ is $Z_2$-symmetric and $\Sigma$ is static, $(\bar{M}_\pm,g_\pm)$ have a common static Killing vector field $\partial_t$ which are tangent to $\Sigma_\pm$ on $\Sigma_\pm$ and satisfies $\psi^*:\partial_t|_{\Sigma_+}\mapsto\partial_t|_{\Sigma_-}$ for the identification of the inner boundaries $\psi:\Sigma_+\to\Sigma_-$.
The lapse functions $\mathcal{N}_\pm$ of $(\bar{M}_\pm,g_\pm)$ and $\mathcal{N}$ of $(\mathcal{M},g,\Sigma)$ with respect to the Killing vector are given by $\mathcal{N}_\pm^2=-g_\pm(\partial_t,\partial_t)$ and $\mathcal{N}^2=-g(\partial_t,\partial_t)$, respectively.
Since the first junction condition is satisfied from the $Z_2$-symmetry, $\mathcal{N}^2=-g(\partial_t,\partial_t)=-\frac{1}{2}[g_+(\partial_t,\partial_t)+g_-(\partial_t,\partial_t)]=-g_\pm(\partial_t,\partial_t)=\mathcal{N}_\pm^2$.
Then we assume that $\mathcal{N}_\pm$ are constant along $\Sigma_\pm$ in $\bar{M}_\pm$ and therefore, $\mathcal{N}=const.$ along $\Sigma$ in $\mathcal{M}$.
This is equivalent to the requirement that the time-time component of the surface stress energy $S_{ab}$ is constant, $S_{tt}=-\epsilon h_{tt}=\epsilon {\mathcal{N}|_{\Sigma}}^2=const.$, since $\epsilon=const.$ from Theorem~\ref{theorem:psf-wh}.
Although the assumption may restrict the class of solutions, the pure-tensional wormholes which have been investigated satisfy it~\cite{barcelo,kokubu}.
We have the following theorem.
\begin{theorem}[Uniqueness of static pure-tensional wormhole spacetimes]
\label{theorem:uniqueness}
Let $(\mathcal{M},g,\Sigma)$ be a four-dimensional static pure-tensional wormhole constructed from $(M_\pm,g_\pm)$.
Let $\epsilon$ and $\theta$ be the tension and the mean curvature of $\Sigma$, respectively.
Assume that the lapse functions $\mathcal{N}_\pm$ of $(\bar{M}_\pm,g_\pm)$ regularly foliate $\bar{M}_\pm$ and are constant along $\Sigma_\pm$.
Then, each of the regions $\bar{M}_\pm$ of $(\mathcal{M},g,\Sigma)$ is isometric to the Schwarzschild spacetime with the mass $m=1/(\sqrt{3}\theta)=-1/(6\sqrt{3}\pi \epsilon)>0$.
\end{theorem}
\begin{proof}
From Definition~\ref{definition:af-ptwh}, $\bar{M}_\pm:=M_\pm\setminus M_\pm^{\rm in}$, which constitute $\mathcal{M}$ by $\mathcal{M}=\bar{M}_+\cup_{\psi,\psi_N}\bar{M}_-$, are the manifolds with the asymptotically flat regions and the inner boundaries $\Sigma_\pm$.
From Theorem~\ref{theorem:psf-wh}, $\Sigma_\pm$ of $(\bar{M}_\pm,g_\pm)$ are photon surfaces with the constant mean curvatures
\begin{equation}
\label{eq:theta-p-uniqueness}
\theta_\pm=\mp6\pi \epsilon.
\end{equation}
Therefore, from the assumptions, $(\bar{M}_\pm,g_\pm)$ are AF-geometrostatic spacetimes regularly foliated by $\mathcal{N}_\pm$ and the inner boundaries $\Sigma_\pm$ are Cederbaum's photon spheres, i.e. they are photon surfaces with the constant lapse functions $\mathcal{N}_\pm$ along them in the AF-geometrostatic spacetimes (Definition~2.6 in~\cite{cederbaum}).
Then, from the uniqueness theorem of photon spheres (Theorem~3.1 in~\cite{cederbaum}), $(\bar{M}_+,g_+)$ is Schwarzschild spacetime with the mass $m=1/(\sqrt{3}\theta_+)>0$.
From Eq.~(\ref{eq:theta-p-uniqueness}), we have
\begin{equation}
m=\frac{1}{\sqrt{3}\theta_+}=-\frac{1}{6\sqrt{3}\pi \epsilon}>0
\end{equation}
and $\epsilon<0$.
Since the $Z_2$-symmetry of $(\mathcal{M},g,\Sigma)$ implies that $(\bar{M}_\pm,g_\pm)$ are isometric, $(\bar{M}_-,g_-)$ is also Schwarzschild spacetime with the mass
\begin{equation}
m=\frac{1}{\sqrt{3}\theta_+}=-\frac{1}{\sqrt{3}\theta_-}=-\frac{1}{6\sqrt{3}\pi \epsilon}>0.
\end{equation}
\end{proof}
Note that the static photon surface in Schwarzschild spacetime is the hypersurface of radius $3m$~\cite{claudel}.

\section{Conclusion}
\label{sec:conclusion}
We defined a joined spacetime (JST), which is obtained by truncating and gluing two spacetimes along the boundaries, in Sec.~\ref{sec:joined-spacetime}.
$Z_2$-symmetry of JSTs was defined in Sec.~\ref{sec:z2-symmetry}.
A pure-tensional JST (PTJST) was defined as a $\Lambda$-vacuum JST with a pure tension shell. 
For a $Z_2$-symmetric pure-tensional joined spacetime (Z2PTJST), we proved that its shell must be photon surfaces of the original spacetimes constituting the JST (Theorem~\ref{theorem:psf-wh}) in Sec.~\ref{sec:psf-wht}.
Conversely, if two isometric $\Lambda$-vacuum spacetimes have photon surfaces, we can join them to give a Z2PTJST.
Therefore, we have solutions of Z2PTJSTs as many as the photon surfaces found in, for example,~\cite{claudel,gibbons_2016,koga3,koga:nonsym-psf}.
\par
Z2PTJSTs have been widely investigated in the contexts of wormholes~\cite{barcelo,kokubu}, baby universes~\cite{barcelo}, and brane worlds~\cite{randall,randall2}.
The shells correspond to the throat in the wormhole cases and the brane we live in in the brane world cases.
One can infer that we can extend Theorem~\ref{theorem:psf-wh} to electrovacuum cases because the coincidence of the shells and photon spheres holds in electrovacuum in~\cite{barcelo,kokubu}.
It is fascinating since the electric charges enrich the variety of thin shell wormhole solutions.
\par
Theorem~\ref{theorem:psf-wh} can be used to deny the possibility to construct Z2PTJST from a given spacetime.
In a stationary axisymmetric spacetime like Kerr spacetime, there can be null circular geodesics, however, photon surfaces would not exist on the radii.
It is because the corotating and counterrotating circular orbits have the different radii in general.
Even if the corotating orbits generate a hypersurface on the radius, the counterrotating orbits cannot be tangent to the surface.
Then the hypersurface does not satisfy the definition of photon surface.
Therefore, we cannot join the two copies of the spacetime along the radii of null circular geodesics to give Z2PTJST.
One needs to violate the $Z_2$-symmetry, the pure tension equation of state of the shell, or $\Lambda$-vacuum condition to construct shell wormholes from axisymmetric spacetimes.
\par
Since the shell of a Z2PTJST coincides with a photon surface, the stability of the JST against the shell perturbation also coincides with the stability against the surface perturbation of the photon surface in the original spacetime.
In Sec.~\ref{sec:stability}, after deriving the master equation for the perturbation of photon surfaces, Eq.~(\ref{eq:stability-another}), from the surface deformation formalism by Capovilla and Guven~\cite{capovilla}, we found its close relationship to the stability of null geodesics on a photon surface introduced in~\cite{koga:psfstability}.
Namely, if null geodesics on a photon surface are stable (unstable), the photon surface itself, and therefore the Z2PTJST, is stable (unstable).
We also confirmed it by solving the perturbation equation for photon surfaces explicitly in Appendix~\ref{app:example} with the specific induced metrics.
\par
It is remarkable that the perturbation equation, Eq.~(\ref{eq:stability-another}), is also useful to seek photon surfaces in the vicinity of a given photon surface.
Actually, we found the hyperboloid Eq.~(\ref{eq:hyperboloid-in-minkowski}) by perturbing the plane of $y=0$ in Minkowski spacetime $\mathbb{M}^3$.
The hyperboloid, as well as the plane, is known to be a timelike photon surface of $\mathbb{M}^3$, indeed~\cite{claudel}.
In contrast with the planar case, in the spherically and hyperbolically symmetric cases, it is suggested that there would not be photon surfaces violating the spatial symmetries because the perturbation $\Phi$ in Eq.~(\ref{eq:Phi-solution-ex}) depends only on the time $t$.
\par
In the wormhole cases of Z2PTJSTs with the vanishing cosmological constant, we applied the uniqueness theorem of photon spheres by Cederbaum~\cite{cederbaum} and established the uniqueness theorem of static pure-tensional wormholes with $Z_2$-symmetry (Theorem~\ref{theorem:uniqueness}) in Sec.~\ref{sec:uniqueness}.
The theorem states that both sides of the wormhole are isometric to Schwarzschild spacetime with the same masses.
It is also interesting that the tension of the shell and the mass of the wormhole are inversely proportional to each other and the positive mass implies the negative tension and the negative energy of the shell.
This is consistent with the result that wormhole spacetimes have to violate energy conditions in the vicinity of the throats~\cite{morris,hochberg}.
\par
The static pure-tensional wormhole in our uniqueness theorem has the photon surface (shell) of the geometry Eq.~(\ref{eq:static-const-curvature-surface-main}) with $\alpha=+1$.
Therefore, from the general solution of the perturbation, Eq.~(\ref{eq:general-pert-alpha-nzero}), in Sec.~\ref{sec:stability}, the static pure-tensional wormhole is unstable against general throat (shell) perturbation preserving $Z_2$-symmetry.

\begin{acknowledgments}
The author thanks T. Kokubu, T. Harada, H. Maeda, M. Kimura, H. Yoshino, Y. Nakayama, Y. Hatsuda, K. Nakashi, and T. Katagiri for their very helpful discussions and comments.
This work was supported by JSPS KAKENHI Grant No. JP19J12007.
\end{acknowledgments}

\appendix
\section{Gaussian normal coordinates on a $Z_2$-symmetric joined spacetime}
\label{app:z2-coords}
We consider a coordinate system of $(\mathcal{M},g,\Sigma)$ which respects the $Z_2$-symmetry.
With the coordinate system, we can verify Definition~\ref{definition:z2-symmetry} explicitly.
\par
Let $(\mathcal{M},g,\Sigma)$ be $Z_2$-symmetric joined spacetime of $(M_\pm,g_\pm)$.
Consider Gaussian normal coordinates $\mathcal{C}_\pm:p_\pm\in M_\pm \mapsto \left(l_\pm,x_\pm\right)\in\mathbb{R}^{d+1}$ with respect to $\Sigma_\pm$ such that
\begin{equation}
g_\pm=dl_\pm^2+h_{ij}^\pm(\pm l_\pm,x_\pm)dx_\pm^idx_\pm^j
\end{equation}
with the conditions,
\begin{eqnarray}
\Sigma_\pm=\left\{l_\pm=0\right\},\\
N_\pm=\partial l_\pm.
\end{eqnarray}
With the isometry $\phi:(M_+,g_+)\to (M_-,g_-)$ in Definition~\ref{definition:z2-symmetry}, we can impose that $\mathcal{C}_-\circ\phi\circ\mathcal{C}_+^{-1}:(l_+,x_+)\mapsto(-l_-,x_-)$ on the coordinates.
It leads to
\begin{equation}
h_{ij}^+(l,x)=h_{ij}^-(l,x)
\end{equation}
for a variable $l$.
Indeed, the assumption on the coordinates satisfies the requirement for $\phi$, Eqs.~(\ref{eq:z2-diffeo-sigma}) and~(\ref{eq:normal-inversion}):
\begin{eqnarray}
\left.\phi\right|_{\Sigma_+}:\Sigma_+=\left\{l_+=0\right\}\to\left\{l_-=0\right\}=\Sigma_-,\\
\phi^*:N_+=\partial l_+\mapsto -\partial l_-=-N_-.
\end{eqnarray}
Since $\mathcal{C}_-\circ\left.\phi\right|_{\Sigma_+}\circ\mathcal{C}_+^{-1}:(0,x_+)\mapsto(0,x_-)$, Eq.~(\ref{eq:z2-identification}) implies that $\mathcal{M}=\bar{M}_+\cup_{\psi,\psi_N}\bar{M}_-$ is obtained by, after the truncation $M_\pm\to \bar{M}_\pm$, identifying the coordinates on $\Sigma_\pm$ as
\begin{equation}
\label{eq:coords-identification-sigma}
\mathcal{C}_-\circ\psi\circ\mathcal{C}_+^{-1}:(0,x_+)\mapsto (0,x_-).
\end{equation}
From the identification of the normal vectors $\psi_N:N_+\mapsto N_-$, Eq.~(\ref{eq:normal-identification}), we also have
\begin{equation}
\label{eq:coords-identification-normal}
\left.\partial l_+\right|_{l_+=0}=\left.\partial l_-\right|_{l_-=0}.
\end{equation}
Then we introduce the coordinate system $\mathcal{C}:p\in\mathcal{M}\mapsto (l,x)$ into $\mathcal{M}$ by
\begin{equation}
(l,x)=\left\{
\begin{array}{ll}
(l_+,x_+) & (l\ge0)\\
(l_-,x_-) & (l<0).
\end{array}
\right. 
\end{equation}
This choice satisfies the conditions for the gluing, Eqs.~(\ref{eq:coords-identification-sigma}) and~(\ref{eq:coords-identification-normal}).
Finally, we obtain the metric distribution on $\mathcal{M}$,
\begin{equation}
g=dl^2+h_{ij}(|l|,x)dx^idx^j,
\end{equation}
where $h_{ij}(l,x):=h^+_{ij}(l,x)=h^-_{ij}(l,x)$ for $l\ge0$.
\par
Obviously, the transformation $l\to -l$ leaves $g$ and $\Sigma$ invariant and exchanges the regions of $\mathcal{M}$ as $M_+^{\rm ex}\leftrightarrow M_-^{\rm ex}$.
Definition~\ref{definition:z2-symmetry} gives a $Z_2$-symmetric joined spacetime, indeed.
\par
The quantities appearing in the junction conditions are given as follows.
From $h_\pm(X,Y)=g_\pm(X,Y)$ $\forall X,Y\in T\Sigma_\pm$, the induced metric is
\begin{equation}
h_+=h_{ij}(0,x)dx^idx^j=h_-.
\end{equation}
From $\chi_\pm=(1/2)\mathcal{L}_{N_\pm}h_\pm$~\cite{textbook:wald}, the second fundamental form is
\begin{equation}
\chi_+=\frac{1}{2}h_{ij,l}(0,x)dx^idx^j=-\chi_-.
\end{equation}
We can easily see that Proposition~\ref{proposition:z2-junction} holds from the expressions.

\section{Calculations in deformation of hypersurfaces}
\label{app:capovilla-recalculation}
We recalculate a part of the calculation in~\cite{capovilla} following their procedure.
The notations in~\cite{capovilla} are converted to ours in the following.
\par
We consider an embedded surface $\Sigma$ of a spacetime $(M,g)$ and its deformation.
The dimension $d:=\dim(\Sigma)\ge1$ is arbitrary here and therefore we have $(D-d)$ unit normal vectors $N^i$ and the deviation scalars $\Phi_i$ where $i=1,...,D-d$.
From Eq.~(3.9) in~\cite{capovilla}, the deformation of the Christoffel symbol $\Gamma_{ab}{}^c$ with respect to the induced metric $h_{ab}$ of the surface $\Sigma$ is given by
\begin{eqnarray}
\label{eq:defo-connection}
\nabla_\delta \Gamma_{ab}{}^c
&=&\frac{1}{2}h^{cd}\left[\nabla^h_a\left(\nabla_\delta h_{bd}\right)+\nabla^h_b\left(\nabla_\delta h_{ad}\right)-\nabla^h_d\left(\nabla_\delta h_{ab}\right)\right]\nonumber\\
&=&h^{cd}\left[\nabla^h_a\left(\chi_{bd}{}^i\Phi_i\right)+\nabla^h_b\left(\chi_{ad}{}^i\Phi_i\right)-\nabla^h_d\left(\chi_{ab}{}^i\Phi_i\right)\right]
\end{eqnarray}
where $\delta:=\Phi_i N^i$, $\nabla$ is the covariant derivative associated with $g$, $\nabla^h$ is the covariant derivative associated with $h_{ab}$, $\chi_{ab}{}^i$ is the $i$-th extrinsic curvature of $\Sigma$ with respect to $N^i$, and $\nabla_\delta:=\delta^\mu \nabla_\mu$.
In general, a variation of a metric $g_{\mu\nu}\to g_{\mu\nu}+\Delta g_{\mu\nu}$ gives the variations $\Delta$ of the connection coefficients $\Gamma^\alpha{}_{\mu\nu}$ and the curvatures $R^\alpha{}_{\mu\beta\nu}$, $R_{\mu\nu}$, and $R$ of a spacetime as~\cite{textbook:poisson,textbook:hawking}
\begin{eqnarray}
\Delta R^\alpha{}_{\mu\beta\nu}
&=&\nabla_\beta\Delta\Gamma^\alpha{}_{\mu\nu}-\nabla_\nu\Delta\Gamma^\alpha{}_{\mu\beta},\\
\Delta R_{\mu\nu}
&=&\nabla_\alpha\Delta\Gamma^\alpha{}_{\mu\nu}-\nabla_\nu\Delta\Gamma^\alpha{}_{\mu\alpha},\\
\Delta R
&=&-R^{\mu\nu}\Delta g_{\mu\nu}+g^{\mu\nu}\left(\nabla_\alpha\Delta\Gamma^\alpha{}_{\mu\nu}-\nabla_\nu\Delta\Gamma^\alpha{}_{\mu\alpha}\right),\\
\Delta \left(\sqrt{-g}R\right)
&=&\frac{1}{2}g^{\mu\nu}\Delta g_{\mu\nu}\sqrt{-g}R+\sqrt{-g}\Delta R.
\end{eqnarray}
Applying the equations to our case and replacing $\Delta$ by $\nabla_\delta$, we have
\begin{eqnarray}
\nabla_\delta\mathcal{R}^a_{cbd}
&=&\nabla^h_b\left(\nabla_\delta\Gamma_{cd}{}^a\right)-\nabla^h_c\left(\nabla_\delta\Gamma_{bd}{}^a\right),\\
\nabla_\delta\mathcal{R}_{cd}
&=&\nabla^h_a \left(\nabla_\delta\Gamma_{cd}{}^a\right)-\nabla^h_d\left( \nabla_\delta\Gamma_{ca}{}^a\right)\\
\nabla_\delta\mathcal{R}
&=&-\mathcal{R}^{cd} \nabla_\delta h_{cd}+h^{cd}\left(\nabla^h_a \nabla_\delta\Gamma_{cd}{}^a-\nabla^h_d \nabla_\delta\Gamma_{ca}{}^a\right),\\
\nabla_\delta \left(\sqrt{-\gamma}\mathcal{R}\right)
&=&\sqrt{-h}\left(\frac{1}{2}h^{ab}\nabla_\delta h_{ab}\mathcal{R}+\nabla_\delta \mathcal{R}\right)
\end{eqnarray}
for the curvatures of $(\Sigma,h)$.
In particular, substituting Eq.~(\ref{eq:defo-connection}) into the identities, we obtain
\begin{eqnarray}
\label{eq:deformation-ricci-scalar}
\nabla_\delta\mathcal{R}
&=&-2\mathcal{R}^{cd} \chi_{cd}{}^i\Phi_i+2\nabla^h_a\nabla^h_c\left(\chi^{aci}\Phi_i\right)-2\Delta^h\left(\theta^{i}\Phi_i\right),\\
\label{eq:deformation-det}
\nabla_\delta \left(\sqrt{-h}\mathcal{R}\right)
&=&-2\sqrt{-h}\mathcal{G}^{ab}\chi_{ab}{}^i\Phi_i+\partial_a\left(\sqrt{-h}J^a\right),
\end{eqnarray}
where we have used the fact $\nabla_\delta h_{ab}=2\chi_{ab}{}^i\Phi_i$ from Eq.~(3.8) in~\cite{capovilla}, $J^a:=2\left(\nabla^h_c\left(\chi^{aci}\Phi_i\right)-{\nabla^h}^a\left(\theta^{i}\Phi_i\right)\right)$ and $\mathcal{G}_{ab}$ is the Einstein tensor of $(\Sigma,h)$.
Note that, modulo a divergence, Eqs.~(\ref{eq:deformation-ricci-scalar}) and~(\ref{eq:deformation-det}) coincide with Eqs.~(3.11) and~(3.12) of~\cite{capovilla}, respectively.

\section{Example: the general solutions of the perturbation equation}
\label{app:example}
The perturbation equation Eq.(\ref{eq:stability-another}),
\begin{equation}
\nabla^h_a\nabla^h_b\Phi-\frac{1}{d}\Delta^h\Phi h_{ab}=\left(\mathcal{R}_{ab}-\frac{\mathcal{R}}{d}h_{ab}\right)\Phi,\nonumber
\end{equation}
for CMC photon surfaces was derived in Sec.~\ref{sec:stability}.
Here we solve the equation explicitly and derive the general solutions in specific cases.
\par
Consider the case where the induced metric $h$ on the photon surface $\Sigma$ is given by
\begin{equation}
\label{eq:static-const-curvature-surface}
h=a^2(-dt^2+\sigma)
\end{equation}
where $a$ is a constant and $\sigma$ is the metric of the $(d-1)$-dimensional space of constant curvature $\alpha=0,\pm1$.
This corresponds to the static cases of~\cite{barcelo,kokubu} with an appropriate scaling of time $t$.
We express the space part as
\begin{equation}
\sigma=d\chi^2+s^2(\chi)\Omega_{d-2}
\end{equation}
with $s(\chi)=\chi,\sin(\chi),\sinh(\chi)$ for $\alpha=0,+1,-1$, respectively, and the metric of the unit $(d-2)$-sphere $\Omega_{d-2}$.
Hereafter we omit the subscript $d-2$ of $\Omega_{d-2}$.
For simplicity, we scale $h$ by a constant so that
\begin{equation}
\label{eq:static-const-curvature-surface-normalized}
h=-dt^2+\sigma
\end{equation}
and solve the perturbation equation, Eq.~(\ref{eq:stability-another}), with this induced metric in the following.
Note that it is sufficient to specify only the intrinsic geometry $(\Sigma,h)$ to solve the equation.
\subsection{$\alpha=\pm 1$}
Consider $\alpha=\pm 1$ cases.
The Ricci curvatures of $(\Sigma,h)$ are given by
\begin{eqnarray}
\mathcal{R}_{tt}=0,\;\;\mathcal{R}_{ti}=0\;\;\mathcal{R}_{ij}=\alpha \sigma_{ij},\;\;\mathcal{R}=(d-1)\alpha
\end{eqnarray} 
where $i,j=\chi,x^A$ and $x^A\;(A=1,...,d-2)$ are the spherical coordinates of $\Omega$.
The double null coordinates $\{\lambda_\pm:=t\pm\chi,x^A\}$ give null geodesics of $(\Sigma,h)$ with the tangents
\begin{equation}
k_\pm=\partial_{\lambda_\pm}.
\end{equation}
From the fact that $\mathcal{R}_{k_\pm k_\pm}=\alpha$ and the contraction of Eq.~(\ref{eq:stability-another}) with $k_\pm$, we have
\begin{equation}
\partial_{\lambda_\pm}^2\Phi(\lambda_+,\lambda_-,x^A)=\frac{1}{4}\alpha\Phi(\lambda_+,\lambda_-,x^A).
\end{equation}
The integration of the ``plus" one of the equation gives
\begin{equation}
\Phi(\lambda_+,\lambda_-,x^A)
=A(\lambda_-,x^A)e^{\sqrt{\alpha}\lambda_+/2}+B(\lambda_-,x^A)e^{-\sqrt{\alpha}\lambda_+/2}
\end{equation}
for the arbitrary functions $A(\lambda_-,x^A)$ and $B(\lambda_-,x^A)$.
Substituting this into the ``minus'' one, we have
\begin{eqnarray}
&&\partial_{\lambda_-}^2A(\lambda_-,x^A)=\frac{1}{4}\alpha A(\lambda_-,x^A),\\
&&\partial_{\lambda_-}^2B(\lambda_-,x^A)=\frac{1}{4}\alpha B(\lambda_-,x^A)
\end{eqnarray}
leading to
\begin{eqnarray}
&&A(\lambda_-,x^A)=C(x^A)e^{\sqrt{\alpha}\lambda_-/2}+D(x^A)e^{-\sqrt{\alpha}\lambda_-/2},\\
&&B(\lambda_-,x^A)=E(x^A)e^{\sqrt{\alpha}\lambda_-/2}+F(x^A)e^{-\sqrt{\alpha}\lambda_-/2}.
\end{eqnarray}
Therefore, we obtain
\begin{eqnarray}
\label{eq:Phi-ansatz}
\Phi&=&C(x^A)e^{\sqrt{\alpha}(\lambda_++\lambda_-)/2}+D(x^A)e^{\sqrt{\alpha}(\lambda_+-\lambda_-)/2}
+E(x^A)e^{-\sqrt{\alpha}(\lambda_+-\lambda_-)/2}+F(x^A)e^{-\sqrt{\alpha}(\lambda_++\lambda_-)/2}\nonumber\\
&=&C(x^A)e^{\sqrt{\alpha}t}+D(x^A)e^{\sqrt{\alpha}\chi}
+E(x^A)e^{-\sqrt{\alpha}\chi}+F(x^A)e^{-\sqrt{\alpha}t}
\end{eqnarray}
for the arbitrary functions $C(x^A)$, $D(x^A)$, $E(x^A)$, and $F(x^A)$.
\par
The Christoffel symbols $\Gamma^a{}_{bc}$ with respect to $h_{ab}$ are calculated as
\begin{eqnarray}
\label{eq:christoffel-h}
&&\Gamma^t{}_{ab}=\Gamma^a{}_{tb}=\Gamma^\chi{}_{\chi\chi}=\Gamma^\chi{}_{\chi A}=0\nonumber\\
&&\Gamma^\chi{}_{AB}=-ss'\Omega_{AB},\;\;\Gamma^A{}_{\chi\chi}=0,\;\;\Gamma^A{}_{\chi B}=\frac{s'}{s}\delta^A_B,\;\;\Gamma^A{}_{BC}={^\Omega\Gamma}^A{}_{BC}
\end{eqnarray}
where ${^\Omega\Gamma}^A{}_{BC}$ is the Christoffel symbol with respect to $\Omega_{AB}$.
The nondiagonal components of Eq.~(\ref{eq:stability-another}) reduce to
\begin{equation}
\label{eq:pert-nondiag}
\partial_a\partial_b\Phi-\Gamma^i{}_{ab}\partial_i\Phi=0\;\;\;(a\neq b).
\end{equation}
For $(a,b)=(t,\chi)$, Eq.~(\ref{eq:pert-nondiag}) gives
\begin{equation}
\partial_t\partial_\chi\Phi=0
\end{equation}
and this is consistent with Eq.~(\ref{eq:Phi-ansatz}).
For $(a,b)=(t,A)$, we have
\begin{equation}
\partial_t\partial_A\Phi=\sqrt{\alpha}\left[C(x^A)_{,A}e^{\sqrt{\alpha}t}-F(x^A)_{,A}e^{-\sqrt{\alpha}t}\right]=0
\end{equation}
by using Eq.~(\ref{eq:Phi-ansatz}).
This implies $C(x^A)_{,A}=F(x^A)_{,A}=0$ and thus they are constant,
\begin{equation}
\label{eq:CF-const}
C(x^A)=C,\;\;\;F(x^A)=F.
\end{equation}
For $(a,b)=(\chi,A)$, we have
\begin{equation}
\partial_\chi\partial_A\Phi-\frac{s'}{s}\partial_A\Phi=
D(x^A)_{,A}\left[\sqrt{\alpha}-\frac{s'}{s}\right]e^{\sqrt{\alpha}\chi}
-E(x^A)_{,A}\left[\sqrt{\alpha}+\frac{s'}{s}\right]e^{-\sqrt{\alpha}\chi}=0.
\end{equation}
This implies $D(x^A)_{,A}=E(x^A)_{,A}=0$ and thus they are constant,
\begin{equation}
\label{eq:DE-const}
D(x^A)=D,\;\;\;E(x^A)=E.
\end{equation}
Now we have
\begin{equation}
\label{eq:Phi-const-coefficients}
\Phi=Ce^{\sqrt{\alpha}t}+De^{\sqrt{\alpha}\chi}
+Ee^{-\sqrt{\alpha}\chi}+Fe^{-\sqrt{\alpha}t}
\end{equation}
from Eqs.~(\ref{eq:Phi-ansatz}),~(\ref{eq:CF-const}), and~(\ref{eq:DE-const}).
\par
From Eq.~(\ref{eq:Phi-const-coefficients}), we have $\partial_A\Phi=0$ and thus,
\begin{equation}
\nabla^h_a\nabla^h_b\Phi=\partial_a\partial_b\Phi-\Gamma^\chi{}_{ab}\partial_\chi\Phi.
\end{equation}
Together with Eq.~(\ref{eq:christoffel-h}), we have
\begin{equation}
\Delta^h\Phi=-\partial_t^2\Phi+\partial_\chi^2\Phi+(d-2)\frac{s'}{s}\partial_\chi\Phi.
\end{equation}
Then the $tt$-, $\chi\chi$-, and $AB$-component of Eq.~(\ref{eq:stability-another}) give
\begin{eqnarray}
\label{eq:phi-diag}
\partial_t^2\Phi+\frac{1}{d}\left[-\partial_t^2\Phi+\partial_\chi^2\Phi+(d-2)\frac{s'}{s}\partial_\chi\Phi\right]
&=&
\frac{d-1}{d}\alpha\Phi,\nonumber\\
\partial_\chi^2\Phi-\frac{1}{d}\left[-\partial_t^2\Phi+\partial_\chi^2\Phi+(d-2)\frac{s'}{s}\partial_\chi\Phi\right]
&=&\frac{1}{d}\alpha\Phi,\nonumber\\
\frac{s'}{s}\partial_\chi\Phi-\frac{1}{d}\left[-\partial_t^2\Phi+\partial_\chi^2\Phi+(d-2)\frac{s'}{s}\partial_\chi\Phi\right]
&=&\frac{1}{d}\alpha\Phi,
\end{eqnarray}
respectively.
The sum of the first and second equations gives
\begin{equation}
\partial_t^2\Phi+\partial_\chi^2\Phi=\alpha \Phi
\end{equation}
and this is already satisfied by Eq.~(\ref{eq:Phi-const-coefficients}).
The subtraction of the second equation times $(d-1)$ from the first equation gives
\begin{equation}
D\left[\sqrt{\alpha} s-s'\right]e^{\sqrt{\alpha}\chi}+E\left[\sqrt{\alpha} s+s'\right]e^{-\sqrt{\alpha}\chi}=0
\end{equation}
with the substitution of Eq.~(\ref{eq:Phi-const-coefficients}).
This requires
\begin{equation}
D=E=0
\end{equation}
and therefore $\partial_\chi\Phi=0$.
The third equation in Eq.~(\ref{eq:phi-diag}) then reduces to
\begin{equation}
\partial_t^2\Phi=\alpha\Phi
\end{equation}
and this is already satisfied by Eq.~(\ref{eq:phi-diag}) with the vanishing of $D$ and $E$. 
\par
As a consequence, the general solution of Eq.~(\ref{eq:stability-another}) with the geometry Eq.~(\ref{eq:static-const-curvature-surface-normalized}) for $\alpha=\pm1$ is given by
\begin{equation}
\Phi=Ce^{\sqrt{\alpha}t}+Fe^{-\sqrt{\alpha}t}
\end{equation}
with the arbitrary constants $C$ and $F$.
\subsection{$\alpha=0$}
Consider the $\alpha=0$ case of the geometry of Eq.~(\ref{eq:static-const-curvature-surface-normalized}).
Since the geometry $(\Sigma,h)$ is the $d$-dimensional Minkowski spacetime, we adopt Cartesian coordinates $\{t,x^i\}$ on it and the curvatures and the Christoffel symbols identically vanish.
Then the master equation Eq.~(\ref{eq:stability-another}) for the perturbation $\Phi$ reduces to
\begin{equation}
\label{eq:stability-flat}
\partial_a\partial_b\Phi-\frac{1}{d}\eta^{cd}\partial_c\partial_d\Phi \eta_{ab}=0.
\end{equation}
From the nondiagonal components, $\Phi$ must be the sum of one-variable functions of $t$ and $x^i$.
We express it as
\begin{equation}
\label{eq:phi-sum}
\Phi=f_t(t)+\sum_{j=1}^{d-1}f_j(x^j)
\end{equation}
with the arbitrary functions $f_t(t)$ and $f_i(x^i)$.
The $tt$- and $ii$-component give
\begin{eqnarray}
(d-1)\partial_t^2\Phi+\sum_{j=1}^{d-1}\partial_j^2\Phi&=&0,\\
\partial_t^2\Phi+d\partial_i^2\Phi-\sum_{j=1}^{d-1}\partial_j^2\Phi&=&0,
\end{eqnarray}
respectively.
The equations give
\begin{equation}
\partial_t^2\Phi+\partial_i^2\Phi=0
\end{equation}
by the summation of them.
From Eq.~(\ref{eq:phi-sum}), this leads to
\begin{equation}
-\partial_t^2f_t(t)=\partial_i^2f_i(x^i)=2D
\end{equation}
with the arbitrary constant $D$.
Integrating the equations, we have the general solution of $\Phi$ for $\alpha=0$,
\begin{equation}
\Phi=D\left[-t^2+\sum_{j=1}^{d-1}(x^j)^2\right]+C_tt+\sum_{j=1}^{d-1}C_jx^j+C,
\end{equation}
with the arbitrary constants $D$, $C_t$, $C_i$, and $C$.
The solution satisfies Eq.~(\ref{eq:stability-flat}), indeed.
\par
The solution can be rewritten as
\begin{equation}
\Phi=D\eta_{ab}x^ax^b+B_ax^a+C
\end{equation}
with $x^a=(t,x^i)$ and $B_a=(C_t,C_i)$.
$C$ represents the perturbation parallel to $\Sigma$, or the shift of $\Sigma$.
It displaces $\Sigma\to\Sigma_\varepsilon$ by a constant distance $C\varepsilon$ at each point $p\in\Sigma$.
$B_a$ rotates $\Sigma$ with the fixed axes $A^a$ given by $B_aA^a=0$.
Points satisfying $x^a\propto A^a$ are fixed by the perturbation and the set of $A^a$ spans $(d-1)$-dimensional surface, actually.
If the surface is timelike, the perturbation is a spatial rotation of $\Sigma$ while if spacelike, it is a Lorentz boost of $\Sigma$.
$D$ provides the perturbation which depends only on the length of $x^a$ and fixes the origin $x^a=0$ and the null rays $\eta_{ab}x^ax^b=0$ passing the origin.
We can understand $D$ as follows.
\par
To imagine the effect of $D$, let us consider that $\Sigma$ is embedded into $(d+1)$-dimensional Minkowski spacetime $\mathbb{M}^{d+1}$ by the embedding $x^a\in\Sigma\mapsto x^\mu=(x^a,y=0)\in M^{d+1}$, for example.
Suppose $B_a=C=0$ for simplicity.
The photon surface $\Sigma$ is given by $y=0$ with the normal vector $N^\mu=(0,1)$.
The perturbed photon surface $\Sigma_\varepsilon$ is given by 
\begin{equation}
x_\varepsilon^\mu=x^\mu+\Phi N^\mu\varepsilon=\left(x^a,D\varepsilon\eta_{ab}x^ax^b\right)
\end{equation}
from Eqs.~(\ref{eq:expansion-epsilon}) and~(\ref{eq:definition-Phi}).
In the case of $d=2$ and $x^a=(t,x)$ for simplicity, it is $x_\varepsilon^\mu=\left(t,x,D\varepsilon(-t^2+x^2)\right)$, or expressed as the quadratic equation,
\begin{equation}
-t^2+x^2-\frac{1}{D\varepsilon}y=0.
\end{equation}
This coincides with the expansion about $y$ of the one-sheeted hyperboloid given by
\begin{equation}
\label{eq:hyperboloid-in-minkowski}
-t^2+x^2+(y-a)^2=a^2
\end{equation}
in $\mathbb{M}^3$ around $y=0$ in linear order, where the ``radius" $a$ is specified by $a=(2D\varepsilon)^{-1}$.
In fact, timelike planes and one-sheeted hyperboloids are known to be timelike photon surfaces of Minkowski spacetime~\cite{claudel}.
Furthermore, by the large radius limit, or equivalently the zero-curvature limit, $a\to\infty$, the local geometry of the hyperboloid approaches to that of planes.
The limit corresponds to $\varepsilon\to0$ in our case.
Therefore, the parameter $D$ gives the perturbation of the plane $\Sigma$ to a hyperboloid $\Sigma_\varepsilon$ of infinitely large radius $a=(2D\varepsilon)^{-1}$.
Note that although the local geometries coincide with each other in the limit, their global topologies, which would be subject to the nonlinear order, in Minkowski spacetime are different.
It is also worth noting that the hyperboloid in Minkowski spacetime has the geometry of de Sitter spacetime~\cite{claudel}.
\subsection{Stability}
The scaling of Eq.~(\ref{eq:static-const-curvature-surface-normalized}) to Eq.~(\ref{eq:static-const-curvature-surface}) gives $\alpha\to a^{-2}\alpha $ and $t\to at$.
As a result, the most general solutions of the linear perturbation $\Phi$ of the photon surface which has the surface geometry of Eq.~(\ref{eq:static-const-curvature-surface}) are
\begin{equation}
\label{eq:Phi-solution-ex}
\Phi=Ce^{\sqrt{\alpha}t}+Fe^{-\sqrt{\alpha}t}\;\;\;(\alpha=\pm 1)
\end{equation}
with the arbitrary constants $C,F$ and
\begin{equation}
\label{eq:Phi-solution-ex-flat}
\Phi=D\eta_{ab}x^ax^b+B_ax^a+C\;\;\;(\alpha=0)
\end{equation}
with the arbitrary constants $D$, $B_a$, and $C$.
The photon surface is stable, i.e. $\Phi$ is bounded, against all the possible perturbations if and only if the spatial geometry is hyperbolically symmetric, $\mathcal{R}=(d-1)\alpha<0$.
The case is where $-C_{kNkN}=\mathcal{R}_{kk}<0$ for any null vector $k\in T_p\Sigma$ at any point $p\in\Sigma$, i.e. $\Sigma$ is a strictly stable photon surface~\cite{koga:psfstability} as we expected.
The $\alpha=0$ case corresponds to a marginally stable case where $-C_{kNkN}=\mathcal{R}_{kk}=0$.
If one perturbs $\Sigma$ with the initial condition $\partial_a\Phi|_{t=t_0}=0$, it leads to $D=B_a=0$ and the deviation remains constant, $\Phi=C$, and is bounded.
It is worth noting that any perturbation violating the spatial symmetry of the surface is not allowed for $\alpha=\pm1$.
The relatively high degrees of freedom of the perturbation for $\alpha=0$ come from that the geometry $(\Sigma,h)$ restores the maximal symmetry on it.
\par
If we apply the results to the perturbation of Z2PTJST in Sec.~\ref{sec:stability}, it is consistent with the $\Lambda$-vacuum case of~\cite{barcelo,kokubu} and implies that the perturbations the authors investigated for the spherically and hyperbolical symmetric cases are the most general under the $Z_2$-symmetry of the JST in the sense of Eq.~(\ref{eq:Phi-solution-ex}).
\par
Eq.~(\ref{eq:cmcpsf-perturbation}) or~(\ref{eq:stability-another}) can shed light on seeking photon surfaces around a given photon surface.
This is because the existence of the possible linear perturbations of a photon surface should imply the existence of nearby photon surfaces.
The result, Eq.~(\ref{eq:Phi-solution-ex}), in the example tells us that, in the vicinity of the spherically and hyperbolically symmetric photon surface of $\Lambda$-vacuum spacetime, there would be no photon surface which does not have the same spatial symmetry.
Eq.~(\ref{eq:Phi-solution-ex-flat}) indicates that there exist photon surfaces around a given planar photon surface.
They are obtained by the shift, the rotation, the boost, and the transformation to hyperboloids.

%
\bibliography{psf_as_wht}

\end{document}